\documentclass[preprint,12pt,number]{article}

\usepackage{amssymb}
\usepackage{amsmath}

\usepackage{lmodern}

\usepackage[utf8]{inputenc}
\usepackage[T1]{fontenc}

\usepackage[english]{babel}
\usepackage{microtype}
\usepackage{csquotes}

\usepackage[x11names, svgnames, dvipsnames]{xcolor}

\usepackage{enumitem}
\setlist[itemize, 1]{label=\textbullet}
\setlist[itemize, 2]{label=$\circ$}
\setlist[itemize, 3]{label={\tiny$\blacksquare$}}

\usepackage{moreenum}

\usepackage{amsmath}
\usepackage{amssymb}
\usepackage{amsthm}

\usepackage{latexsym}
\usepackage{mathtools}
\usepackage{mathrsfs}
\usepackage{dsfont}
\usepackage{xfrac}
\usepackage{cancel}
\usepackage{centernot}
\usepackage{empheq}

\usepackage[sanserif, full, disableredefinitions]{complexity}

\usepackage{stmaryrd}
\usepackage{bbding}
\usepackage{accents}
\usepackage{esint}
\usepackage{pifont}

\usepackage{tabularx}
\usepackage{booktabs}

\input{insbox.tex}

\usepackage{caption}

\usepackage{environ}

\usepackage{geometry}
\newgeometry{vmargin={3cm,3cm}, hmargin={3cm,3cm}, marginparwidth=2.5cm}

\usepackage{knowledge}

\usepackage[symbol=$\uparrow\;$, numberlinked=false]{footnotebackref}

\usepackage{hyperref}
\hypersetup{
    plainpages=false,
    linktoc=section,
	colorlinks=true,
	breaklinks=true,
	pageanchor=true,
    citecolor=Green,
    filecolor=Mulberry,
    menucolor=Burgundy,
    runcolor=Mulberry,
    urlcolor=Blue,
	anchorcolor=black,
	pdftitle={Higher order logic and enumeration},
	pdfauthor={},
}

\usepackage{graphicx}

\newcommand{\longv}[1]{}

\usepackage{todonotes}

 \newtheorem{theorem}{Theorem}[section]
 \newtheorem{lemma}[theorem]{Lemma}
 \newtheorem{corollary}[theorem]{Corollary}
 
\newtheorem{proposition}[theorem]{Proposition}

\newtheorem{defn}[theorem]{Definition}

 \newtheorem{example}[theorem]{Example}
 \newtheorem{remark}[theorem]{Remark}

\theoremstyle{definition}
  \newtheorem{notation}[theorem]{Notation}
  
\usepackage[framemethod=tikz]{mdframed}

\surroundwithmdframed[linecolor=gray,linewidth=2pt,leftline=true,topline=false,bottomline=false,rightline=false,fontcolor=gray]{objective}

\newcommand\N{\mathbb{N}}

\newcommand{\cp}[1]{\mathbf{#1}}
 
\renewcommand{\NC}{\cp{NC}}
\renewcommand{\AC}{\cp{FAC}}

\newcommand{\FNC}{\mathbf{FNC}}

\newcommand{\kBRN}{{ k\text{-}\rm BRN}}
\newcommand{\CRN}{{\rm CRN}}
\newcommand{\WBRN}{{\rm WBRN}}

\newcommand{\rank}[1]{\mathsf{rank}(#1)}

\newcommand{\Nat}{\mathbb{N}}

\newcommand{\sucs}{\mathbf{s}}

\newcommand{\tu}[1]{\mathbf{#1}}

 \newcommand\lengthnotation{\ell}
\newcommand{\length}[1]{\mathrm{\lengthnotation}(#1)}

\longv{

}

\newcommand{\LDL}{\mathbb{LDL}}

\newcommand{\ACDL}{\mathbb{ACDL}}

\newcommand{\TCDL}{\mathbb{TCDL}}

\newcommand{\maj}{\textsc{Maj}}

\newcommand{\FPTime}{\mathbf{FP}}
\newcommand{\Dlogtime}{\mathbf{Dlogtime}}
\newcommand{\ACz}{\mathbf{AC}^0}
\newcommand{\ACi}{\mathbf{AC}^i}
\newcommand{\NCi}{\mathbf{NC}^i}
\newcommand{\NCc}{\mathbf{NC}}
\newcommand{\TCz}{\mathbf{TC}^0}
\newcommand{\TCi}{\mathbf{TC}^i}

\newcommand{\FTC}{\mathbf{FTC}}
\newcommand{\FLH}{\mathbf{FLH}}

\newcommand{\fun}[1]{\mathsf{#1}}

\newcommand{\raus}[1]{}

\begin{document}




\title{Towards New Characterizations of Small Circuit Classes via Discrete Ordinary Differential Equations} 






\author{Melissa Antonelli \quad Arnaud Durand \quad Juha Kontinen}



\maketitle

\begin{abstract}
Implicit computational complexity is a lively area of theoretical computer science, which aims to provide machine-independent characterizations of relevant complexity classes.
%
One of the seminal works in this field appeared in the 1960s, when Cobham introduced a function algebra closed under bounded recursion on notation to capture polynomial time computable functions ($\FPTime$).
Later on, several complexity classes have been characterized using \emph{limited} recursion schemas.
In this context, an original approach has been recently introduced, showing that ordinary differential equations (ODEs) offer a natural tool for algorithmic design and providing a characterization of $\FPTime$ by a new ODE-schema.
In the present paper we generalize this approach by presenting original ODE-characterizations for the small circuit classes $\AC^0$ and $\FTC^0$.
\end{abstract}








\section{Introduction}
By now, computability and complexity theory have a long history, dating back to pioneering studies by G\"odel~\cite{Godel}, Turing~\cite{Turing}, Church and Kleene~\cite{ChurchKleene} and to connected early developments of algorithmic analysis.
In general, computability theory investigates the limits of what is algorithmically computable, whereas complexity theory classifies functions based on the amount of resources, typically time and space, needed to execute them.
Historically, time and space complexity for Turing machines were first introduced in an explicit context~\cite{HartmanisStearn}, as associated with a machine and a cost model. 
More recently, a different approach to complexity, called implicit computational complexity (ICC, for short), has been developed, aiming to describe complexity classes without explicitly referring to a specific machine model.
This alternative perspective has offered remarkable insights into the corresponding classes and has led to related meta-theorems in several domains, from database theory to constraint satisfaction.

In particular, from the 1990s on, several implicit characterizations for multiple classes have been introduced. 
Just to quote a few approaches among the plethora developed in different fields, in recursion theory classes are captured by limiting recursion schemas in multiple ways~\cite{Cobham,Ritchie,BellantoniCook}, in model theory and descriptive complexity classes are characterized in terms of the richness of the logical language needed to describe the corresponding problem~\cite{Immerman99}, and in proof theory many achievements rely on the so-called Curry-Howard correspondence~\cite{CurryFeys,Howard,SorensenUrzyczyn,Buss}.
The results presented in this paper are placed in the first area.
Indeed, one of the major approaches to computability and complexity theory is constituted by the study of recursion. 
A foundational work in this area is by Cobham~\cite{Cobham}, who provided an implicit characterization of the class of polynomial time computable functions, $\FPTime$, by \emph{bounded recursion on notation} (BRN, for short).
This result has led to several other characterizations based on \emph{limited} recursion schemas, and, for $\FPTime$ itself, alternative approaches have been developed, for instance~relying on safe recursion~\cite{BellantoniCook} and ramification~\cite{Leivant,LeivantMarion}.\footnote{Other implicit characterizations based on schemas and restrictions of first-order programs have been recently introduced in~\cite{BonfanteMarionPechoux,Leivant,LeivantMarion2000,MarionPechoux,NigglWunderlich}.}
Among them there is the  machine-independent approach recently introduced by Bournez and Durand in 2019, which is based on discrete Ordinary Differential Equations (ODEs, for short)~\cite{BournezDurand19}.
Our work takes its cue from this result.

\bigskip
In the context of the present paper, it is particularly relevant to notice that Cobham's work, together with other early results in recursion theory~\cite{Grzegorczyk,Bennett,Ritchie,Lind}, also paved the way to recursion-theoretic characterizations for small circuit classes.
%
In~\cite{Clote88,Clote1990}, Clote introduced an algebra based on so-called \emph{concatenation recursion on notation} (CRN, for short) to capture functions computable in ($\Dlogtime$-uniform) $\ACz$, i.e.~computable by families of polynomial size and constant depth circuits.
This result was extended to $\ACi$ and $\NCi$ due to the notions of \emph{k-bounded} and \emph{weak bounded recursion on notation} in~\cite{Clote1990}.
A few years later, a similar function algebra for capturing $\TCz$ was introduced~\cite{CloteTakeuti}.
Other characterizations for subclasses of $\NCc$ were independently presented in~\cite{ComptonLaflamme,Allen,BonfanteKahleMarionOitavem}.

Related, but alternative, approaches to capture small circuit classes have also been developed in the framework of model- and proof-theory.
In the former area, it is well known that there is an equivalence between $\ACz$ and first-order logic, which naturally generalizes to extensions of this logic and larger circuit classes~\cite{Immerman87,BarringtonImmermanStraubing,HellaKontinenLuosto,GurevichLewis,Lindell}.
In the proof-theoretical setting, in~\cite{ComptonLaflamme} it is shown that functions computable in $\NCc^1$ can be characterized not only by means of a function algebra $\mathscr{A}$, based on so-called \emph{upward tree recursion}, but also in terms of a logical system.
Similarly, in~\cite{Allen}, together with the corresponding algebra, Allen defined a proof system \emph{\'a la Buss} to capture $\NCc$.
In~\cite{CloteTakeuti92}, alternative bounded theories were introduced to capture this and other small circuit classes~\cite{CloteTakeuti}.
Theories for $\TCz$ have been developed in~\cite{Johannsen,CloteTakeuti} and, scaling up to the second-order, in~\cite{NguyenCook06}.
Alternative, proof-theoretical characterizations for $\NC^1$ were then presented in~\cite{Arai}, and, in the context of two-sorted theories, in~\cite{CookMorioka}.

\bigskip
As anticipated, an original approach to complexity, based on discrete ODEs, was introduced in~\cite{BournezDurand19,BournezDurand23},
aiming to capture functions computable in a given complexity class as solutions of a corresponding type of ODE.
In this vain, a characterization of $\FPTime$ was given by \emph{linear} systems of equations and \emph{derivation along a logarithmically growing function}.
Informally, the latter condition controls the number of steps, whereas linearity controls the growth of objects generated during the computation.
Recently, these approaches have been generalized to the continuous setting~\cite{BlancBournez22,BlancBournez23}.

Although small circuit classes have been characterized in multiple ways, the questions of whether they can be studied through ODE lenses and whether this would shed some new light on their features are still open.
To us, these questions are both interesting and challenging. 
They are \emph{interesting} as, for a descriptive approach based on ODEs to make sense and be fruitful, it has to be able to cope with very subtle and restricted modes of computation. 
They are \emph{challenging} as even simple and useful mathematical functions may not be computable in the classes we are considering (e.g.~multiplication is not in $\AC^0$, but it is in $\FTC^0\subseteq \FNC^1$); consequently, tools at hand and the naturalness of the approach are drastically restricted.
Our project aims to investigate these questions and to find natural ODE-oriented function algebras to capture small circuit classes.
So far, we have focused on the characterization of functions computable by families of polynomial size and constant depth circuits ($\AC^0$, for short), possibly including majority gates ($\FTC^0$, for short).
We have captured \emph{both} classes by means of special ODE-schemas, obtained by deriving along the logarithmic function and intuitively allowing for bit shifting operations through restricted forms of linear equations.
These case studies are intended as the first steps towards a \emph{uniform} characterization of other relevant classes in the $\AC^i$ and $\FNC^i$ hierarchies.

\paragraph{The Structure of This Work}
This article is the extended version of the works presented in the eponymous~\cite{ADK24b} and in~\cite{ADK24a}.
The paper is structured as follows.
In Section~\ref{sec:preliminaries}, we introduce notions and results at the basis of our ODE-style characterizations. In particular, in Section~\ref{sec:CCviaODE}, we summarize the salient aspects of the approach developed in~\cite{BournezDurand23}, which we aim to generalize from the study of $\FPTime$ to that of small circuit classes. 
In Section~\ref{sec:parallel}, we briefly recap some basic notions in parallel complexity and recall the definition of a couple of function algebras developed in this setting.
Then, in Section~\ref{sec:main}, we present our ODE-characterizations for $\AC^0$ and $\FTC^0$. 
Specifically, in Section~\ref{sec:schemas}, we introduce new restricted ODE-schemas intuitively corresponding to variously left- and right-shifting the binary representation of a given input.
Based on these schemas, in Sections~\ref{sec:FAC} and~\ref{sec:FTC} we provide different ODE-style characterizations for $\AC^0$ and $\FTC^0$ (respectively).
To clarify the relationship between the ODE machinery and constant-depth circuit computation, and to identify the minimal computational principles necessary to capture these classes, several schemas and characterizations are presented.
In Section~\ref{sec:direct}, we provide alternative and direct completeness proofs for both these classes in the non-uniform setting, namely assuming that functions describing the circuits are given as basic functions.
Finally, in Section~\ref{sec:conclusion}, we conclude by pointing to possible directions for future research.

\section{Preliminaries}\label{sec:preliminaries}
In this Section we introduce the preliminary notions needed to develop our characterization of $\AC^0$ and $\FTC^0$.
In particular, in Section~\ref{sec:CCviaODE}, we briefly present the elegant approach to complexity delineated in~\cite{BournezDurand19,BournezDurand23}, while, in Section~\ref{sec:parallel}, we recap salient definitions related to small circuit classes and the state-of-the-art on function algebras to capture them.

\subsection{Capturing Complexity via ODEs}\label{sec:CCviaODE}

In this Section we summarize the core ideas of the $\FPTime$ characterization presented in~\cite{BournezDurand19,BournezDurand23}, that we will extend to circuit classes in Section~\ref{sec:main}. 
This approach combines discrete-based computation (as digital computers are discrete machines) and ODEs, allowing to capture $\FPTime$ due to a special discrete ODE-schema. 
%
%
Concretely, we start by briefly recapping Cobham's result, capturing $\FPTime$ by a schema which imposes an explicit bound on the recursion in the form of an already defined function (Section~\ref{sec:fromCtoCC}).
Then, the alternative characterization based on ODEs is introduced (Section~\ref{sec:fromBRtoODEs}).
%

\subsubsection{From Computability to Complexity}\label{sec:fromCtoCC}

Modern computing machines have been developed in relation to the first attempts to formalize the notion of algorithm and of computable functions.
Accordingly, computational models are inherently linked to function algebras, where a function algebra is an algebra constituted by a set of basic functions and by schemas to construct new functions from the given ones.
For example, the class of primitive recursive functions is the smallest class defined as follows:
$$
\mathcal{PR} = [\fun{0}, \fun{s}, \pi^p_i; \circ, \text{PR}]
$$
where $\fun{s}(x)=x+1$ is the successor function, $\pi^p_i$ is the projection function, $\circ$ is function composition and PR is the primitive recursion schema below.

\begin{defn}[Primitive Recursion, PR]\label{def:PR}
    Given $g:\Nat^p \to \Nat$ and $h:\Nat^{p+1}\to \Nat$, the function $f:\Nat^{p+1}\to \Nat$ is defined by \emph{primitive recursion} from $g$ and $h$ if:
    \begin{align*}
        f(0,\tu y) &= g(\tu y) \\
        f(x+1, \tu y) &= h\big(x, \tu y, f(x, \tu y)\big).
    \end{align*}
\end{defn}
\noindent
Clearly, primitive recursive computation requires a number of steps that is approximately exponential in the length of the input. However, when the recursion schema is restricted in the proper way, several other classes could be defined, including complexity ones.

In 1943, Kalm\'ar defined the class of elementary functions by restricting PR so to allow only limited sums and products~\cite{Kalmar}, and, in 1953, Grzegorczyk generalized this work to the study of the classes $\mathcal{E}^k$. 
About ten years later, Bennett introduced the notion of bounded recursion on notation in the context of rudimentary sets~\cite{Bennett}.
Then, a groundbreaking result in this spirit was established by Cobham, who provided a characterization of $\FPTime$ using the notion of BRN, as defined below, and inspired by Bennett's recursion~\cite{Bennett}.

\begin{defn}[Bounded Recursion on Notation, BRN]\label{defn:BRN}
   Given $g:\Nat^p \to \Nat$, $k: \Nat^{p+1} \to \Nat$ and $h_i : \Nat^{p+2} \to \Nat$, with $i\in \{0,1\}$, a function $f:\Nat^{p+1} \to \Nat$ is defined by \emph{bounded recursion on notation} from $g,h_0, h_1$ and $k$ if
   \begin{align*}
       f(0,\tu y) &= g(\tu y) \\
       f(\fun{s}_i(x), \tu y) &= h_i\big(x, \tu y, f(x, \tu y)\big) \quad x \neq 0 \\
       f(x,\tu y) &\leq k(x,\tu y)
   \end{align*}
   where $\fun{s}_0 = 2x$ and $\fun{s}_1 = 2x+1$ are the binary successor functions.
\end{defn}
Observe that if $\ell(k(x,\tu y))$ is polynomial in $\ell(x) + \ell(\tu y)$, then so is $\ell(f(x,\tu y))$ (being $\ell(z)$ the length of $z$ written in binary, see Notation~\ref{not:ell}).
This ensures that inner terms do not grow too fast, i.e.~that the growth rate is controlled by means of already known functions.
On the other hand, the number of induction steps is kept under control by the binary successor functions, as $\ell(\fun{s}_1(x))=\ell(\fun{s}_0(x))=\ell(x)+1$.
It was relying on the BRN schema, that Cobham defined the algebra below:\footnote{Actually, Cobham's algebra is made of functions on words, while we are here presenting its number version, see~\cite{Rose}.}
$$
\mathcal{A}_{Cob} = [\fun{0}, \fun{s}_0, \fun{s}_1, \pi^p_i, \#; \circ, \text{BRN}],
$$
where $x\# y=2^{\ell(x)\times \ell(y)}$.
It was proven that this function algebra precisely captures the class of poly-time computable functions, i.e.
$$
\mathcal{A}_{Cob} = \FPTime.\footnote{In~\cite{Cobham}, the displayed algebra is claimed to characterize $\FPTime$ but no complete proof is presented. A formal proof is offered in~\cite{Rose} (where $\pi^n_k$ is denoted by $i^n_k$, $\#$ by $\overline{D}$, and BRN is called limited recursion on notation (LRN)). An alternative but equivalent definition of $\mathcal{A}_{Cob}$, based on limited iteration on notation, is introduced in~\cite{Buss}.} 
$$
This result has led to several machine-independent characterizations of other complexity classes, even of parallel ones (see Section~\ref{sec:parallel}).

\subsubsection{From Bounded Recursion to ODEs}\label{sec:fromBRtoODEs}

Cobham's work not only led to a variety of new, implicit characterizations for different classes, but also inspired alternative characterizations for $\FPTime$, such as those based on safe recursion~\cite{BellantoniCook} and ramification~\cite{Leivant,LeivantMarion}.
Among them, the proposal presented in~\cite{BournezDurand23} offers the advantages of not imposing an explicit bound on the recursion schema and of not assigning specific roles to variables.
These nice properties are obtained by introducing two original notions in the context of discrete ODEs: the idea of \emph{deriving along specific functions}, so that, if the function is properly chosen, the number of computation steps is kept under control, and a special syntactic form for equations, called \emph{essential linearity} and allowing to control the object size.
It was relying on these two features that poly-time computability is captured.

\paragraph{Discrete ODE Schemas}
Classical ODEs have been widely used in applied sciences, serving as a model in several contexts involving continuous phenomena.
When a change occurs incrementally rather than continuously, their counterpart is defined by \emph{Discrete ODEs} or \emph{difference equations}, which are, in turn, capable of modeling discrete changes.
%
%
An obvious observation is that a discrete function can be naturally described by its right \emph{discrete derivative} defined as~$\Delta \tu f(x)=\tu f(x+1)-\tu f(x)$ (also denoted $\tu f'(x)$ in what follows), and that, in this setting, it is also straightforward to rewrite primitive recursion.
%


A discrete calculus analogous to differential calculus has been developed (see, e.g.,~\cite{Spiegel}), and many techniques and results from classical derivatives generalize to the context of discrete derivatives.
The notion of discrete integral is defined so that, for a function $\tu f(x)$, 
$
\int^b_a \tu f(x)\delta x = \sum^{x=b-1}_{x=a} \tu f(x)
$
with the convention that $\int^b_a \tu f(x)\delta x = 0$ when $a=b$ and $\int^b_a \tu f(x)\delta x=-\int^b_a \tu f(x)\delta x$, when $a>b$.
The fundamental theorem for the finite calculus states that, given the function $\tu F(x)$,
$$
\int^b_a \tu F'(x) \delta x = \tu F(b)- \tu F(a).
$$

Given the analogy between finite and classical calculus, it is also natural to consider a theory of discrete differential equations corresponding to the traditional theory of differential equations.
Discrete ODEs are expressions of the (possibly vectorial) form
$$
\frac{\partial \tu f(x,\tu y)}{\partial x} = \tu h\big( x, \tu y, \tu f(x, \tu y)\big)
$$
where $\frac{\partial \tu f(x,\tu y)}{\partial x}$ stands for the derivative of $\tu f(x,\tu y)$ considered as a function of $x$, when $\tu y$ is fixed, i.e.
$$
\frac{\partial \tu f(x,\tu y)}{\partial x} = \tu f(x+1, \tu y) - \tu f(x,\tu y).
$$
The concept of solution is as expected.
When some initial value $\tu f(0,\tu y)=\tu g(\tu y)$ is added, this is called an \emph{Initial Value Problem} (IVP) or \emph{Cauchy Problem}.
An IVP can always be put in integral form
and, if it is
defined as above, it always admits a necessarily unique solution over $\Nat$.

\begin{remark}
    Since coping with negative numbers on classical models of computation can be done through simple encodings, as in~\cite{BournezDurand19}, we will restrict our treatment to non-negative numbers.
\end{remark}

As anticipated, discrete ODEs are convenient tools to define functions.

\begin{defn}[Discrete ODE Schema]
    Given the functions $\tu g$
    and $\tu h$, 
    we say that $\tu f$ 
    is defined by \emph{discrete ODE solving from $g$ and $h$}, if $\tu f$ is the (necessarily unique) solution of the IVP:
    \begin{align*}
        \tu f(0,\tu y) &= \tu g(\tu y) \\
        \frac{\partial \tu f(x,\tu y)}{\partial x} &= \tu h\big(x, \tu y, \tu f(x, \tu y)\big).
    \end{align*}
\end{defn}
\noindent
The class of functions that are defined using discrete ODEs instead of primitive recursion is easily seen to correspond to primitive recursive functions.
On the other hand, in order to capture fine-grained complexity measures, some restrictions are needed.

Particularly relevant and studied 
are discrete linear or affine ODEs, namely restricted discrete ODEs of the form:
$$
\tu f'(x) = \tu A(x) \cdot \tu f(x) + \tu B(x),
$$
where $\tu A$ is a matrix and $\tu B$ a vector.
The corresponding schema is defined as follows:

\begin{defn}[Linear ODE Schema]
    Given vectors $\tu G=(G_i)_{1\leq i\leq k}, \tu B=(B_i)_{1\leq i\leq k}$, and matrix $\tu A= (A_{i,j})_{1\leq i, j\leq k}$, we say that $\tu f$ is obtained by \emph{linear ODE solving from $\tu G, \tu A$ and $\tu B$}, if $\tu f:\Nat^{p+1} \to \mathbb{Z}^k$ corresponds to the (necessarily unique) solution of the IVP:
    \begin{align*}
        \tu f(0,\tu y) &= \tu G(\tu y) \\
        \frac{\partial \tu f(x,\tu y)}{\partial x} &= \tu A(x,\tu y) \times \tu f(x,\tu y) + \tu B(x, \tu y).
    \end{align*}
\end{defn}
\noindent
In~\cite[Lemma 3.15]{BournezDurand23}, it is shown that, for matrix $\tu A$ and vectors $\tu B$ and $\tu G$, the solution of the equation 
$$
\tu f'(x,\tu y) = \tu A\big(x, \tu y, \tu f(x,\tu y), \tu h(x,\tu y)\big) \times \tu f(x,\tu y) + \tu B\big(x, \tu y, \tu f(x, \tu y), \tu h(x, \tu y)\big)
$$
with initial value $\tu f(0,\tu y)=\tu G(\tu y)$ is
$$
\sum^{x-1}_{u=-1}\bigg(\prod^{x-1}_{t=u+1}\Big(1+\tu A\big(t, \tu y, \tu f(t,\tu y), \tu h(t,\tu y)\big)\Big)\bigg) \times \tu B\big(u, \tu y, \tu f(u, \tu y), \tu h(u,\tu y)\big)
$$
with the conventions that $\tu B(\cdot, \cdot, -1, \tu y)=\tu G(\tu y)$ and $\prod^{x-1}_x\kappa (x)=1$.
Remarkably, this can be interpreted as an algorithm unrolling the computation of $\tu f(x+1,\tu y)$ from the computation of $\tu f(x,\tu y), \tu f(x-1, \tu y), \dots, \tu f(0,\tu y)$ in a dynamic programming way.


\paragraph{On Length ODE}
It is now possible to introduce the idea of deriving along a given function.

\begin{notation}\label{not:ell}
    Let $\ell(x)$ denote the length of $x$ written in binary, i.e. $\big\lceil \text{log}_2(x+1)\big\rceil$, so that, in particular, $\ell(0)=0$ and $\ell(1)=1$.
    Let $\ell_2$ be a shorthand for $\ell\circ \ell$, that is $\ell_2(x)=\ell(\ell(x))$ and, more generally for $i\ge 2$, $\ell_i(x)=\ell(\ell_{i-1}(x))$.
\end{notation}

\begin{defn}[$\lambda$-ODE Schema]
    Let $\tu f, \lambda$ and $\tu h$
    be functions.
    Then,
    \begin{align}
    \frac{\partial \tu f(x,\tu y)}{\partial \lambda} = \frac{\partial \tu f(x,\tu y)}{\partial \lambda (x,\tu y)} = \tu h\big(x, \tu y, \tu f(x,\tu y)\big)
    \end{align}
    is a formal synonym of $\tu f(x+1, \tu y)=\tu f(x, \tu y) + \big(\lambda (x+1, \tu y) - \lambda (x,\tu y)\big)\times \tu h\big(x, \tu y, \tu f(x,\tu y)\big)$.
    When $\lambda (x,\tu y)=\ell(x)$, we call (1) a \emph{length}-ODE.
\end{defn}
\noindent
Intuitively, one of the key properties of the $\lambda$-ODE schema is its dependence on the number of distinct values taken by $\lambda$, that is, the value of $\tu f(x,\tu y)$ changes only when the value of $\lambda(x,\tu y)$ does.

The computation of the solutions of $\lambda$-ODEs has been fully described in~\cite{BournezDurand23}.
Here, we focus on the special cases of $\lambda=\ell$ and $\lambda =\ell_2$, which are particularly relevant to characterize circuit classes.
%
Clearly, the value of $\ell(x)$ changes (it increases by 1) when $x$ goes from $2^t-1$ to $2^t$.
The value of $\ell_2(x)$ changes when $x$ goes from $2^{2^t-1}-1$ to $2^{2^t-1}$.

\begin{notation}~\label{notation: alpha functions}
    For $u\ge 0$, let us denote by $\alpha(u)=2^u-1$ the greatest integer whose length is $u$; similarly, $\alpha_2(u)=2^{2^u-1}-1$ is the greatest integer $t$ such that $\ell_2(t)=u$.
\end{notation}
\noindent
If $\tu f$ is a solution of Equation (1) with $\lambda =\ell$ and initial value $\tu f(0,\tu y)=\tu g(\tu y)$, then:
$$
\tu f(1, \tu y)= \tu f(0,\tu y) + \tu h\big(\alpha(0), \tu y, \tu f(\alpha(0), \tu y)\big).
$$
More generally, for all $x$ and $\tu y$:
$$
    \tu f(x,\tu y) = \tu f(x-1, \tu y) + \Delta(\ell(x-1)) \times \tu h\big(x-1, \tu y, \tu f(x-1, \tu y)\big) 
$$

\noindent where $\Delta(\ell(t-1))=\ell(t)-\ell(t-1)$. Hence, by induction, for any $t\leq x$:
\begin{align*}
    \tu f(x,\tu y)
    &= \begin{cases}
    \tu f(t, \tu y) \quad \quad &\text{if } \ell(x)-\ell(t)=0 \\
    \tu f(t, \tu y) + \tu h(t, \tu y, \tu f(t, \tu y)) \quad \quad  &\text{if } \ell(x)-\ell(t) = 1 
    \end{cases} \\
    &= \tu f\big(\alpha(\ell(x)-1), \tu y\big) + \tu h\big(\alpha (\ell(x)-1), \tu y, \tu f(\alpha (\ell(x)-1), \tu y)\big).
\end{align*}
 Indeed, $\alpha (\ell(x)-1)$ is the largest integer whose length is different from $\ell(x)$ (in fact equal to $\ell(x)-1$, and consequently, $\Delta(\alpha (\ell(x)-1))=1$). For every $t$ strictly in between $\alpha (\ell(x)-1)$ and $x$, it holds that $\Delta(\ell(t))=0$ and $ \tu f(x,\tu y)= \tu f(t,\tu y)$. 
Starting from $t=x\ge 1$ and taking successive decreasing values of $t$, the first difference such that $\Delta(t)\neq 0$ is defined by the biggest $t-1$ such that $\ell(t-1)=\ell(x)-1$, i.e.~$t-1=\alpha(\ell(x)-1)$.
Hence, by induction it is shown that:
$$
\tu f(x,\tu y) = \sum^{\ell(x)-1}_{u=-1} \tu h\big(\alpha (u), \tu y, \tu f(\alpha(u), \tu y)\big),
$$
with $\tu h(\alpha(-1), \tu y, \cdot)=\tu f(0,\tu y)$ and, as seen, $\alpha(u)=2^u -1$.
Similarly, if $\tu f$ is a solution of Equation (1) for $\lambda = \ell_2$ with the same initial value, then, for all $x$ and $\tu y$:
$$
\tu f(x,\tu y) = \sum^{\ell_2-1}_{u=-1} \tu h\big( \alpha_2(u), \tu y, \tu f(\alpha_2(u), \tu y)\big),
$$
with $\tu h(\alpha_2(-1), \tu y, \cdot)=\tu f(0,\tu y)$ and $\alpha_2(u)=2^{2^u-1}-1$.

\paragraph{On Linear Length ODE}
%
A second fundamental concept is that of \emph{(essential) linearity}, adapted so to control the growth of functions defined by ODEs.
To introduce it, we start by formally presenting the notion of \emph{degree for a polynomial expression}, which is a slightly generalized version of the corresponding definition in~\cite{BournezDurand23}.
Although we will mainly use it in the context of expressions of one variable, here we present the notion of degree of an expression and linearity in full generality, i.e.~in relation to a set of variables. 

We first formally introduce the sign (and cosign) function(s) we will use in the following and returning 1 when the argument is strictly positive and 0 otherwise. 

\begin{defn}[Sign Function]
Let $\fun{sg}:\mathbb{Z} \to \mathbb{Z}$ be the sign function over $\mathbb{Z}$, taking value 1 for $x>0$ and 0 otherwise.
Let $\fun{cosg}:\mathbb{Z} \to \mathbb{Z}$ be the cosign function over $\mathbb{Z}$ defined as $\fun{cosg}(x)=1-\fun{sg}(x)$.
\end{defn}

\noindent
It is then possible to introduce the notions of $\fun{sg}$-polynomial expression and of corresponding degree.  

\begin{defn}[Degree of Expression]\label{def:degree}
    A  $\fun{sg}$-\emph{polynomial expression} is an expression built over the signature $\{+,-,\times\}$, the $\fun{sg}$ function and a set of variables $X=\{x_1,\dots, x_h\}$ plus integer constants.
    %
    %
    Given a set of variables $\tu x = \{x_{i_1},\dots, x_{i_m}\} \subseteq X$,
    the \emph{degree of the set in a $\fun{sg}$-polynomial expression $P$}, $\deg(\tu x, P)$, is defined as:
    \begin{itemize}
        \itemsep0em
        \item $\deg(\tu x, P)=0$ for $P$ constant
        \item $\deg(\tu x, x_{i_j})=1$, for $x_{i_j}\in \{x_{i_1},\dots, x_{i_m}\}$ and $\deg(\tu x, x_{i_j})=0$ for $x_{i_j} \not\in \{x_{i_1},\dots, x_{i_m}\}$
        \item $\deg(\tu x, Q_1+Q_2)=\deg(\tu x, Q_1-Q_2)=\max\{\deg(\tu x,Q_1), \deg(\tu x, Q_2)\}$
        \item $\deg(\tu x, Q_1\times Q_2)=\deg(\tu x, Q_1)+ \deg(\tu x, Q_2)$
        \item $\deg(\tu x, \fun{sg}(P))=0$.
    \end{itemize}
\end{defn}
\noindent
For readability, given $x_1\in X$, we denote $\deg(\{x_1\},P)$ by $\deg(x_1,P)$.\footnote{Notice that \cite[Def. 5.14]{BournezDurand23} becomes a special case of Definition~\ref{def:degree}, namely,~the one dealing with a single variable only.}

\begin{example}
    Let $\tu x = \{x_1, x_2, x_3\}$ and $P'=3 \times x_1 \fun{sg}(x_3)+ 2x_2 \times x_1$.
    Then, $\deg(x_1,P')=\deg(x_2,P')=1$ and $\deg(x_3,P')=0$, whereas $\deg(\tu x,P')=2$.
\end{example}

\begin{defn}[Essentially Constant Expression]
    A $\fun{sg}$-polynomial expression $P$ is \emph{essentially constant in a set of variables $\tu x$} when $\deg(\tu x,P)=0$.
\end{defn}
\noindent

\begin{defn}[Essentially Linear Expression]
    A $\fun{sg}$-polynomial expression $P$ is \emph{essentially linear} in a set of variables $\tu x$ when $\deg(\tu x,P)=1$.
\end{defn}
\noindent
So, when dealing with a single variable $x$, a $\fun{sg}$-polynomial expression $P$ is said to be essentially linear in it if 
there exist $\fun{sg}$-polynomial expressions $Q_1$ and $Q_2$ such that $P=Q_1\times x + Q_2$ and $\deg(x,Q_1)=\deg(x,Q_2)=0$ (i.e.,~$x$ is essentially constant in $P$).

\noindent 
With a slight abuse of notation, we say that, for a given set of variables $\tu x$, a $\fun{sg}$-polynomial expression $P$ is essentially constant (resp., linear) in $\tu f(\tu x)$, seen as a new set of variables,  when $\deg( \tu f( \tu x), P)=0$ (resp., $\deg( \tu f(\tu x),P)=1$).

\begin{example}
    Let $\tu x=\{x_1, x_2, x_3\}$.
    The expression $P(x_1,x_2,x_3)=3x_1\times x_3 + 2x_2\times x_3$ is essentially linear in $x_1$, in $x_2$ and in $x_3$. It is not essentially linear in  $\tu x$, as $\deg(P,\tu x)=2$.
    %
    %
    The expression $P'(x_1,x_2,x_3)=x_1 \times \fun{sg}((x_1-x_3) \times x_2)+x_2^3$ is essentially linear in $x_1$, essentially constant in $x_3$ and not linear in $x_2$.
     Clearly, $P'$ is not linear in $\tu x$.
     The expression $Q(f_1(x_1,x_2), f_2(x_1,x_2))=f_1(x_1,x_2)\times f_2(x_1,x_2)$ is linear in $f_1(x_1,x_2)$ and in $f_2(x_1,x_2)$ but not in $\tu f(x_1,x_2)=\{f_1(x_1,x_2),f_2(x_1,x_2)\}$.
\end{example}

It is now possible to introduce the desired linear $\lambda$-ODE schema, capturing poly-time computation.

\begin{defn}[Linear $\lambda$-ODE]
Given the functions $\tu g, \tu h, \lambda$ and $\tu u$, the function $\tu f$ is
    \emph{linear $\lambda$-ODE definable from $\tu g, \tu h$ and $\tu u$} if it is the solution of the IVP:
    \begin{align*}
        \tu f(0,\tu y) &= \tu g(\tu y) \\
        \frac{\partial \tu f(x,\tu y)}{\partial \lambda} &= \tu u\big(x, \tu y, \tu f(x,\tu y), \tu h(x,\tu y)\big),
    \end{align*}
    where $\tu u$ is essentially linear in $\tu f(x,\tu y)$.
    When $\lambda$ is the length function $\ell$, such schema is called \emph{linear length ODE}, $\ell$-\emph{ODE}.
\end{defn}

\noindent
We can even rephrase this definition stating that there exist $\tu A$ and $\tu B$ which are $\fun{sg}$-polynomial expressions essentially constant in $\tu f(x,\tu y)$ and such that:
$$
    \frac{\partial \tu f(x,\tu y)}{\partial \ell} = \tu A\big(x, \tu y, \tu f(x,\tu y), \tu h(x,\tu y)\big) \times \tu f(x,\tu y) + \tu B\big(x, \tu y, \tu f(x,\tu y), \tu h(x, \tu y)\big).
$$
Following~\cite{BournezDurand23}, when deriving along $\ell$, for all $x$ and $\tu y$, $\tu f(x,\tu y)$ is equal to:
$$
\sum^{\ell(x)-1}_{u=-1} \Bigg(\prod^{\ell(x)-1}_{t=u+1} \Big(1+\tu A\big(\alpha(t), \tu y, \tu f(\alpha(t), \tu y), \tu h(\alpha(t), \tu y)\big)\Big)\Bigg) \times \tu B\big(\alpha(t), \tu y, \tu f(\alpha(u), \tu y), \tu h(\alpha(u), \tu y)\big)
$$
with the convention that $\prod^{x-1}_x \kappa(x)=1$ and $\tu B(\cdot, \cdot, \alpha(-1), \tu y)=\tu f(0,\tu y)$.

\begin{example}[Function $2^{\ell(x)}$]\label{ex:2ell}
    The function $x\mapsto 2^{\ell(x)}$ can be seen as the solution of the IVP:
    \begin{align*}
        f(0) &= 1 \\
        \frac{\partial f(x)}{\partial \ell} &= f(x)
    \end{align*}
    that is a special case of $\ell$-\emph{ODE} such that $\tu G=1$, $\tu A=1$ and $\tu B=0$.
    Indeed, the solution of this system is of the form:
    $$
    f(x) = \prod^{\ell(x)-1}_{t=0} 2= 2^{\ell(x)}.
    $$
    In addition, the function $(x,y):x,y \mapsto 2^{\ell(x)} \times y$ can be captured by the same equation with initial value $f(0,y)=y$. 
    Indeed, $f(x,y)=\prod^{\ell(x)-1}_{t=0} 2 \times y =2^{\ell(x)}\times y$.
\end{example}

\paragraph{The Class $\LDL$}
The core of \cite{BournezDurand23} is the proof that the $\ell$-ODE schema actually captures poly-time computation~\cite[Cor. 5.21]{BournezDurand23}.
Relying on this result, Theorem 6.7 provides an original, implicit characterization of $\FPTime$ in terms of a function algebra made of basic functions $\fun{0}, \fun{1}, \pi^p_i, +, -, \times, \ell, \fun{sg}$ and closed under composition $(\circ)$ and $\ell$-ODE: 
$$
\LDL = [\fun{0}, \fun{1}, \ell, \fun{sg}, +, -, \times, \ell, \pi^p_i; \circ, \ell\text{-ODE}].
$$

\subsection{On Parallel Computation}\label{sec:parallel}

In this Section, we recap basic notions in parallel complexity so to introduce small circuit classes (Section~\ref{sec:parCC}) and briefly outline the implicit characterizations of two of them as presented in the literature (Section~\ref{sec:art}).

\subsubsection{Boolean Circuits and Parallel  Complexity classes}\label{sec:parCC}

We consider parallel complexity measures defined by Boolean circuit computation.
A Boolean circuit is a vertex-labeled directed acyclic graph whose nodes are either input nodes (no incoming edges), output nodes (no outgoing edges) or nodes labeled with a Boolean function in a given set.
%
The set of functions allowed in a Boolean circuit is given by its \emph{basis}, and the set $\{\neg, \wedge,\vee\}$ forms a \emph{universal basis}.
A bounded fan-in Boolean circuit is a circuit over the (standard) basis $\{\neg, \wedge^2, \vee^2\}$, while an unbounded fan-in circuit is defined over the (standard) basis $\{\neg, (\wedge^n)_{n\in \Nat}, (\vee^n)_{n\in \Nat}\}$.
A Boolean circuit with majority gates allows in addition gates labeled by the function $\maj$, that outputs 1 when the majority of its inputs are 1's.
A family of circuits ($C_n$) is $\Dlogtime$-uniform if there is a TM (with a random access tape) that decides in deterministic logarithmic time the \emph{direct connection language of the circuit}, i.e.~such that, given $1^n, a, b$ and $t\in \{\wedge,\vee, \neg\}$, it decides if $a$ is of type $t$ and $b$ is a predecessor of $a$ in the circuit (and analogously for input and output nodes).

When dealing with circuits, some resources of interests are \emph{size}, i.e.~the number of its gates, and \emph{depth}, i.e.~the length of the longest path from the input to the output (see~\cite{Vollmer} for more details and related results), as well as the choice of bounding or not the fan-in of gates, in particular for small classes.
In this paper we will focus on two small circuit classes of functions, namely $\AC^0$ and $\FTC^0$, which are the first levels of the hierarchies $\AC$ and $\FTC$.

\begin{defn}
    For $i\in \Nat$, the class $\ACi$ (resp., $\TCi$) is the class of languages recognized by a $\Dlogtime$-uniform family of Boolean circuits (resp., circuits including majority gates) of polynomial size and depth $O((log$ $n)^i)$.
    We denote by $\AC^i$ and $\FTC^i$ the corresponding function classes.
\end{defn}
\noindent
In other words, $\AC^0$ (resp., $\FTC^0$) is the class of functions recognized by a family of unbounded circuits (including $\maj$) of polynomial size and constant depth. We set $\AC=\bigcup_{i\in \N}\AC^i$ and $\FTC=\bigcup_{i\in \N}\FTC^i$. We also denote by $\FNC=\bigcup_{i\in \N}\FNC^i$ the similar hierarchy for bounded fan-in circuits.


\subsubsection{Function Algebras for Small Parallel Classes}\label{sec:art}

As mentioned, Cobham's work paved the way to several recursion-theoretic characterizations for classes other than $\FPTime$, including small circuit ones.
In particular, in 1988/90, Clote introduced the algebra $\mathcal{A}_0$ to capture functions in the log-time hierarchy (and equivalent to $\AC^0$)~\cite{Clote88,Clote1990}.
To do so, he introduced so-called \emph{concatenation recursion on notation} (CRN, for short), as inspired by a similar schema developed by Lind in the context of logarithmic space complexity~\cite{Lind}.
Then, this work was extended to other parallel classes, $\AC^i$ and $\FNC$, due to the notions of \emph{k-bounded recursion on notation} ($k$-BRN, for short) and \emph{weak bounded recursion on notation} (WBRN, for short).
Indeed, by introducing the algebras $\mathcal{A}_i$ and $\mathcal{A}$, Clote captured $\AC^i$ and $\FNC$, resp.~\cite{Clote1990}.
In the same year, also Compton and Laflamme presented an algebra (and a logic) to characterize $\FNC^1$~\cite{ComptonLaflamme}.
Their algebra is based on a recursion schema called \emph{upward tree recursion}.
In 1991, an alternative characterization of $\NC$ was provided by Allen, who also introduced a corresponding arithmetic theory in the style of Buss~\cite{Allen}.
The most important schema in his algebra was so-called \emph{polynomially bounded branching recursion}, which abstracts the divide-and-conquer technique.
In 1992 and 1995, other small circuit classes were considered (and re-considered) by Clote and Takeuti, who also introduced bounded theories for them~\cite{CloteTakeuti92,CloteTakeuti}.
In particular, in~\cite{CloteTakeuti} a function algebra to capture $\FTC^0$ was presented.
Another recent recursion-theoretic characterization of $\NC^i$ was developed by Bonfante, Khale, Marion and Oitavem, this time following the ramification approach by Leivant, Bellantoni and Cook~\cite{BonfanteKahleMarionOitavem}.

Since they are of particular interest for our work, we present a few details on Clote and Takeuti's characterizations of $\AC^0$ and $\FTC^0$. 
%
To this aim, we need to formally introduce a recursion schema called concatenation recursion on notation and obtained by limiting the BRN one.

\begin{defn}[Concatenation Recursion on Notation, CRN]
    A function $f$ is defined by \emph{concatenation recursion on notation} from $g,h_0$ and $h_1$, denoted by $f=CRN(g,h_0,h_1)$ if, for all $x$ and $\tu y$:
    \begin{align*}
        f(0,\tu y) &= g(\tu y) \\
        f(\fun{s}_0(x), \tu y) &= \fun{s}_{h_0(x,\tu y)}\big(f(x,\tu y)\big) \quad for \; x \neq 0 \\
        f(\fun{s}_1(x), \tu y) &= \fun{s}_{h_1(x,\tu y)}\big(f(x,\tu y)\big)
    \end{align*}
    with $h_i\in \{0,1\}$.
\end{defn}

%
%
%
%
\noindent
It was relying on this schemas that the function algebras $\mathcal{A}_0$ and $\mathcal{TC}_0$ were introduced and recursion-theoretic characterizations for $\AC^0$ and $\FTC^0$ were obtained.\footnote{We stick to the original notation from~\cite{Clote1990,CloteTakeuti92,CloteTakeuti}.}

\begin{defn}[Clote's Class $\mathcal{A}_0$]
    Let $\mathcal{A}_0$ be the algebra:
    $$
    \mathcal{A}_0 = [\fun{0}, \fun{s}_0, \fun{s}_1, \ell, \fun{BIT}, \#, \pi^p_i; \circ, \emph{CRN}]
    $$
    where $\fun{BIT}(x,y) = \big\lfloor \frac{x}{2^y}\big\rfloor\mod 2$ returns the value of the $y$th bit in the binary representation of $x$ and $x\# y=2^{\ell(x) \times \ell(y)}$ is the smash function.
\end{defn}
\noindent
This class was proved able to capture the logarithmic time hierarchy $\FLH$, which is equal to $\AC^0$.

\begin{theorem}[\cite{Clote88,Clote1990}]
    $\AC^0=\FLH=\mathcal{A}_0$.
\end{theorem}
\noindent
That $\mathcal{A}_0 \subseteq \FLH$ is established by showing that basic functions are computable in log-time and that $\FLH$ is closed under composition and CRN;
that $\FLH \subseteq \mathcal{A}_0$ is proved by arithmetizing the computation of a log-time bounded random access machines.
\\
In~\cite{CloteTakeuti}, these results are generalized to $\FTC^0$, which is proved equivalent to the function algebra below:
$$
\mathcal{TC}_0 = [\fun{0}, \fun{s}_0, \fun{s}_1, \ell, \fun{BIT}, \times, \#, \pi^p_i; \circ, \text{CRN}].
$$


\section{Towards an ODE Characterization of $\AC^0$}\label{sec:main}

In this Section, we provide the first implicit characterization of $\AC^0$ in the ODE setting.
We start by introducing the new ODE schemas which are at the basis of our characterizations of $\AC^0$ and $\FTC^0$ and intuitively correspond to (variously) left- and right-shifting (Section~\ref{sec:schemas}). 
Relying on them, we introduce the function algebra $\ACDL$, the defining feature of which is precisely the presence of ODE schemas, and prove that this algebra captures $\AC^0$ (Section~\ref{sec:FAC}).
This is established indirectly, by~passing through Clote's $\mathcal{A}_0$. 
Aiming to fully clarify the relationship between our discrete ODE schemas and $\AC^0$, as well as the ``minimal'' computational principles necessary to capture this class, alternative ODE-based characterizations for $\AC^0$ are presented.
As an additional consequence of these results, similar characterizations for $\FTC^0$ are naturally provided (Section~\ref{sec:FTC}).
Finally, alternative, direct proofs of completeness are presented for both classes in a non-uniform setting (Section~\ref{sec:direct}).

\begin{remark}
    Here, functions can take images in $\mathbb{Z}$.
    Accordingly, a convention for the binary representation of integers must be adopted, e.g.~by assuming that, in any binary sequence, the first bit indicates the sign.
    Then, all algorithmic operations can be easily re-designed to handle the encoding of possibly negative integers by circuits of the same size and depth.
\end{remark}

\subsection{Discrete ODE-Schemas for Shifting}\label{sec:schemas}

Let us start by introducing a few preliminary notions.

\begin{notation}
    In the following we will use $\div 2$ to denote integer division by 2, i.e.~for all $x\in \mathbb{Z}$, $x\div 2 =\big\lfloor \frac{x}{2}\big\rfloor$.
\end{notation}
\noindent
To deal with $\AC^0$, we also introduce the auxiliary notions of weak arithmetic expression and of $\fun{sg}$-weak arithmetic expression, which are the natural counterparts of the concepts of polynomial and $\fun{sg}$-polynomial expressions from~\cite{BournezDurand23}, respectively, in the $\AC^0$ setting.

\begin{defn}[Weak Arithmetic Expression]
    A \emph{weak (arithmetic) expression} is an expression built over the signature $\{+,-,\div 2\}$.
    A \emph{$\fun{sg}$-weak (arithmetic) expression} is an expression built over the signature $\{+,-, \div 2, \fun{sg}\}$.
\end{defn}
\noindent
When considering  $\fun{sg}$-expressions, we assume that in them the function $\fun{sg}$ is used only in weak, non-signed expressions, i.e.~no imbrication of $\fun{sg}$ is allowed.

\begin{notation}
    In the following, we will write $A(x,\tu y, \tu f, \tu h)$ as a shorthand for the expression $A\big(x, \tu y, \tu f(x,\tu y), \tu h(x,\tu y)\big)$, $A(x, \tu y, \tu h)$ as a shorthand for $A(x, \tu y, \tu h(x,\tu y))$ and $A(x,\tu y)$ when the expression is ``purely arithmetic'', i.e.~does not contain any call to $\tu f(x,\tu y)$ or $\tu h(x,\tu y)$.
\end{notation}

\begin{remark}
    If $f_1,\dots, f_k$ are functions computable in $\AC^0$, then any weak arithmetic expression $A(x, \tu y, f_1,\dots, f_k)$ is computable in $\AC^0$.
\end{remark}

We now introduce the ODE schemas which are at the basis of our characterizations of $\AC^0$ and $\FTC^0$.
Observe that they will sometimes include $\times$.
This is admissible since, as we shall see, the ``kind of multiplication'' we consider is actually limited to special cases (namely, multiplication by $2^i$), which are proved to be computable in $\AC^0$.

\paragraph{The $\ell$-ODE$_1$ and $\ell$-ODE$_2$ Schemas}
We start with the limited $\ell$-ODE$_1$ schema, intuitively corresponding to basic left shifting(s) and (possibly) adding a bit.

\begin{defn}[Schema $\ell$-ODE$_1$]\label{def:ODE1}
    Given $g:\Nat^p \to \Nat$ and $h:\Nat^{p+1} \to \Nat$, such that $h$ takes values in $\{0,1\}$ only, the function $f:\Nat^{p+1} \to \Nat$ is defined by $\ell$-\emph{ODE}$_1$ from functions $g$ and $h$ when it is the solution of the IVP:
    \begin{align*}
        f(0,\tu y) &= g(\tu y) \\
        \frac{\partial f(x,\tu y)}{\partial \ell} &= f(x,\tu y) + h(x, \tu y).
    \end{align*}
\end{defn}
\begin{remark}
    An equivalent, purely-syntactical formulation of Definition~\ref{def:ODE1} is obtained by substituting the explicit constraint imposing that $h(x,\tu y)\in \{0,1\}$ with the assumption that, in the more general linear form of the equation, such that $\frac{\partial f(x,\tu y)}{\partial \ell}= A(f, h, x, \tu y) \times f(x, \tu y) + B(f, h, x, y)$, again $A=1$ but $B$ is of the form $\fun{sg}\big(B'(h_1,\dots, h_m,x,\tu y)\big)$, where $B'$ is a weak arithmetic expression, calling previously defined $\AC^0$ functions $h_1(x,\tu y),\dots, h_m(x,\tu y)$.
\end{remark}
\noindent
The definition of the function $2^{\ell(x)}$ (and $2^{\ell(x)}\times y$) presented in Example~\ref{ex:2ell} is an instance of $\ell$-ODE$_1$, corresponding to the special case of $g(\tu y)=1$ and $h(x,\tu y)=0$.

\begin{remark}
    An equivalent, purely-syntactical formulation of Definition~\ref{def:ODE1} is obtained by substituting the explicit constraint imposing that $h(x,\tu y)\in \{0,1\}$ with the assumption that it is of the form $\fun{sg}\big(B(h_1,\dots, h_m,x,\tu y)\big)$, where B is a $\fun{sg}$-weak arithmetic expression, calling previously defined $\AC^0$ functions $h_1(x,\tu y),\dots, h_m(x,\tu y)$.
\end{remark}

In terms of circuits, this schema intuitively allows us to iteratively left-shifting (the binary representation of) a given number, each time possibly adding 1 to its final position.
This is clarified by the proof below, establishing that $\AC^0$ is closed under the mentioned schema.

\begin{proposition}\label{prop:ODE1}
    If $f$ is defined by $\ell$-\emph{ODE}$_1$ from $g$ and $h$ in $\AC^0$, then $f$ is in $\AC^0$.
\end{proposition}

\begin{proof}
By Definition~\ref{def:ODE1}, for all $x$ and $\tu y$:
\begin{align*}
    f(x,\tu y) &= \sum^{\ell(x)-1}_{u=-1} \Bigg(\prod^{\ell(x)-1}_{t=u+1} 2\Bigg) \times h\big(\alpha(u), \tu y\big) \\
    &= \sum^{\ell(x)-1}_{u=-1} 2^{\ell(x)-u-1} \times h\big(\alpha(u), \tu y\big)
\end{align*}
with the convention that $\alpha(u)=2^u-1$, for any $x$ and $\kappa(x)$, $\prod^{x-1}_x \kappa(x)=1$ and $h(\alpha(-1), \tu y)=f(0,\tu y)$.
Notice that the given multiplication is always by a power of 2 decreasing for each increasing value of $u$, which basically corresponds to left-shifting (so can be computed in $\AC^0$).
Hence, since by Definition~\ref{def:ODE1}, $h(x,\tu y)\in \{0,1\}$, the outermost sum amounts to a concatenation (which again can be computed in $\AC^0$).

Concretely, for any inputs $x$ and $\tu y$, the desired polynomial-sized and constant-depth circuit to compute $f(x,\tu y)$ is defined as follows:
\begin{itemize}
    \itemsep0em
    \item \begin{sloppypar} In parallel, compute the values of $g(\tu y)$ and of each $h(\alpha(u), \tu y)$, with $u\in \{0,\dots, \ell(x)-1\}$.
    For hypothesis, $g$ and $h$ are computable in $\AC^0$, and, since there are $\ell(x)+1$ initial values to be computed, the entire desired computation can be done in polynomial size and constant depth.
    \end{sloppypar}
    \item In one step, (left-)shift the binary representation of $h(\alpha(-1), \tu y)=g(\tu y)$ by padding $\ell(x)$ zeros on the right and, for $u\ge 0$, (left-)shift each value $h(\alpha(u), \tu y)$ by padding on the right $\ell(x)-u-1$ zeros (this corresponds to multiply by $2^{\ell(x)-u-1}$) and padding on the left $u+\ell(g(\tu y))$ zeros.
    \item Compute bit by bit the disjunction of all values computed above.
    Clearly, this is done in constant depth.
\end{itemize}
\end{proof}
\noindent
Observe that this schema is not as weak as it may seem as, together with the $\fun{sg}$ function, it is already sufficient to express bounded quantification.

\begin{remark}\label{remark:boundedQ}
    Let $R\subseteq \Nat^{p+1}$ and $h_R$ be its characteristic function.
    Then, for all $x$ and $\tu y$, it holds that $(\exists z\leq \ell(x))R(z,\tu y)=\fun{sg}(f(x,\tu y))$, where $f$ is the solution of the IVP:
    \begin{align*}
        f(0,\tu y) &= h_R(0,\tu y) \\
        \frac{\partial f(x,\tu y)}{\partial \ell} &= f(x,\tu y) + h_R\big(\ell(x+1), \tu y\big)
    \end{align*}
    which is an instance of $\ell$-\emph{ODE}$_1$.
    Intuitively,
    $f(x,\tu y)\neq 0$ when, for some $z$ smaller than $\ell(x)$, $R(z,\tu y)$ is satisfied (i.e.~$h_R(z,\tu y)=1$): if such instance exists, our bounded search ends with a positive answer.
    Clearly, this is conceptually the same as the treatment of bounded quantification by Clote (see~\cite{CloteKranakis}), although in his function algebra $\fun{sg}$ is not primitive.
    
    Universally bounded quantification can be expressed in a similar way. 
    Let $\fun{cosg}(x)=1-\fun{sg}(x)$ and consider $\fun{cosg}(f(x,\tu y))$, for $f$ defined substituting the value of $h_R$ with its co-sign. So, $(\forall z\leq \ell(x))R(z,\tu y)=\fun{cosg}(f(x,\tu y))$ for $f$ such that:
    \begin{align*}
        f(0,\tu y) &= \fun{cosg}\big(h_R(0,\tu y)\big) \\
        \frac{\partial f(x,\tu y)}{\partial \ell} &= f(x,\tu y) + \fun{cosg}\big(h_R(\ell(x+1), \tu y)\big).
    \end{align*}
\end{remark}

We now consider a more general schema, called $\ell$-ODE$_2$.
Intuitively it allows multiple left-shifting, such that each ``basic operation'' corresponds to shifting a given value \emph{of a specific number of digits}, determined by $2^{\ell(k(\tu y))}$.

\begin{defn}[Schema $\ell$-ODE$_2$]
    Given $g:\Nat^p \to \Nat$, $h: \Nat^{p+1} \to \Nat$ and $k:\Nat^p \to \Nat$, the function $f:\Nat^{p+1} \to \Nat$ is defined by $\ell$-\emph{ODE}$_2$ from functions $g, h$ and $k$ if it is the solution of the IVP:
    \begin{align*}
        f(0,\tu y) &= g(\tu y) \\
        \frac{\partial f(x,\tu y)}{\partial \ell} &= \big(2^{\ell(k(\tu y))}-1\big) \times f(x,\tu y) + h(x,\tu y),
    \end{align*}
    where $h(x,\tu y) \in \{0,1\}$ and if, for some $x$ and $\tu y$, $h(x, \tu y)= 1$, then $k(\tu y)\neq 0$.
\end{defn}
\noindent
Again, we are using the multiplication symbol with a slight abuse of notation as the ``multiplication'' involved in the computation above is of a special kind, i.e.~corresponds to multiplying by $2^i$, which is computable in constant depth.
Since this schema is introduced to characterize $\AC^0$, the constraint imposing $k(\tu y)\neq 0$, when at some point $h(x,\tu y)$ takes value 1, is really essential.
Indeed, as we shall see (Section~\ref{sec:FTC}), if we omit it, $\ell$-ODE$_2$ will be too strong, as able to capture binary counting, which is not in $\AC^0$.

Observe that $\ell$-ODE$_1$ is a special case of $\ell$-ODE$_2$, such that $\ell(k(\tu y))=1$ and that, also for it, the desired closure property holds.

\begin{proposition}\label{prop:ODE2}
    If $f$ is defined by $\ell$-\emph{ODE}$_2$ from $g$ and $h$ in $\AC^0$, then $f$ is in $\AC^0$.
\end{proposition}
\begin{proof}
    There are two main cases to be taken into account.
    If $k(\tu y)\neq 0$ the proof is similar to that of Proposition~\ref{prop:ODE1}.
    Indeed, for all $x$ and $\tu y$:
    \begin{align*}
        f(x,\tu y) &= \sum^{\ell(x)-1}_{u=-1} \Bigg(\prod^{\ell(x)-1}_{t=u+1} 2^{\ell(k(\tu y))}\Bigg) \times h\big(\alpha(u), \tu y\big) \\
        &= \sum^{\ell(x)-1}_{u=-1} 2^{\ell(k(\tu y))\times (\ell(x)-u-1)} \times h\big(\alpha(u), \tu y\big)
    \end{align*}
    with the convention that $\alpha(u)=2^u-1$, $\prod^{x-1}_x \kappa(x)=1$ and $h(\alpha(-1), \tu y) = f(0,\tu y)$.
    As said, multiplication here corresponds to a left-shifting such that its ``basic shifting'' corresponds to a left movement of $\ell(k(\tu y))$ digits.
    As for the basic case, it can be easily shown that this operation can be implemented by a constant-depth circuit.
    Then, analogously to Proposition~\ref{prop:ODE1}, the outermost iterated sum amounts to concatenation (as, by construction, $h(\alpha(u), \tu y)\in \{0,1\}$).

    Concretely, we can construct a constant-depth circuit generalizing the procedure defined for the special case of $\ell$-ODE$_1$:
    \begin{itemize}
        \itemsep0em
        \item In parallel, compute the values of $g(\tu y)$ and $h(\alpha(u), \tu y)$, for each $u\in \{0,\dots, \ell(x)-1\}$. 
        This can be done in constant depth by hypothesis.
        \item In one step, shift the value of $g(\tu y)$ by padding $\ell(k(\tu y)) \times \ell(x)$ zeros on the right, and, for $u\ge 0$, shift all values $h(\alpha(u), \tu y)$ by padding on the right $\ell(k(\tu y)) \times (\ell(x)-u-1)$ zeros (i.e.~multiplying by $2^{\ell(k(\tu y))\times (\ell(x)-u-1)}$) and by padding on the left $\ell(g(\tu y)) + \ell(k(\tu y)) \times (u+1)-1$ zeros.
        \item Compute bit by bit the disjunction of all the values above.
    \end{itemize}
    In the special case of $k(\tu y)=0$ and $h(x,\tu y)=0$, for all $x$ and $\tu y$, it holds that $f(x,\tu y)=g(\tu y)$.
    This is clearly computable in $\AC^0$, as corresponding to computing $g(\tu y)$, which is computable in constant depth by hypothesis.
\end{proof}

Let us also introduce a weaker form of $\ell$-ODE$_2$, called $wk$-ODE$_2$.
Intuitively, this schema is defined by simply omitting the ``bit addition'' part of $\ell$-ODE$_2$.
As we shall see in Section~\ref{sec:FAC}, together with $\ell$-ODE$_1$ (and forthcoming $\ell$-ODE$_3$), this schema is expressive enough to capture $\AC^0$.

\begin{defn}[Schema $wk$-ODE$_2$]
    Given $g:\Nat^p \to \Nat$, the function $f:\Nat^{p+1}\to \Nat$ is defined by wk-\emph{ODE}$_2$ from $g$ if it is the solution of the IVP:
    \begin{align*}
        f(0,\tu y) &= g(\tu y) \\
        \frac{\partial f(x,\tu y)}{\partial \ell} &= \big(2^{\ell(\tu y)}-1\big) \times f(x,\tu y).
    \end{align*}
\end{defn}


\paragraph{The $\ell$-ODE$_3$ Schema}
Let us now consider $\ell$-ODE$_3$, intuitively corresponding to (basic) right-shifting operations.

\begin{defn}[Schema $\ell$-ODE$_3$]\label{def:ODE3}
    Given $g:\Nat^p\to \Nat$, the function $f:\Nat^{p+1}\to \Nat$ is defined by $\ell$-\emph{ODE}$_3$ from $g$ if it is the solution of the IVP:
    \begin{align*}
        f(0,\tu y) &= g(\tu y) \\
        \frac{\partial f(x,\tu y)}{\partial \ell} &= - \Bigg\lceil \frac{f(x,\tu y)}{2}\Bigg\rceil
    \end{align*}
    where $\lceil \frac{z}{2}\rceil$ is a shorthand for $z-(z\div 2)$.
\end{defn}
\noindent
For $x> 0$, the above equation can be rewritten as:
$$
f(x,\tu y) = f(x-1, \tu y) - \Delta \ell(x-1) \times \Bigg\lceil  \frac{f(x-1,\tu y)}{2} \Bigg\rceil
$$
where, as seen, $\Delta\ell(x-1)=\ell(x)-\ell(x-1)$.
Observe that, also in this case, we are using $\times$ with a slight abuse of notation: indeed, we are dealing with ``bit multiplication'' and multiplying a number by 0 or 1 can be easily done in $\AC^0$ (and re-written in our setting using the basic conditional function, which, as we will see in Section~\ref{sec:FAC}, can be defined in $\ACDL$). 
In other words,
\begin{align*}
    f(x,\tu y) &= \begin{cases}
        f(x-1,\tu y) \quad \quad &\text{if } \ell(x)=\ell(x-1) \\
        f(x-1, \tu y) - \Big\lceil \frac{f(x-1,\tu y)}{2} \Big\rceil \quad \quad &\text{otherwise}
    \end{cases} \\
    &= \begin{cases}
        f(x-1, \tu y) \quad \quad \quad \quad \quad \quad
        \quad \; &\text{if } \ell(x) = \ell(x-1) \\
        \Big\lfloor \frac{f(x-1,\tu y)}{2} \Big\rfloor \quad \quad &\text{otherwise}
    \end{cases} \\
    &=
    \begin{cases}
        f(x-1,\tu y) \quad \quad \quad \quad \quad \quad \quad \quad &\text{if } \ell(x) = \ell(x-1) \\
        f(x-1, \tu y) \div 2 \quad \quad &\text{otherwise.}
    \end{cases}
    \end{align*}
A bit more formally,
    $$
f(x,\tu y) 
= \Bigg\lfloor \frac{f(\beta(\ell(x)-1),\tu y)}{2} \Bigg\rfloor
= \Bigg\lfloor \frac{f(2^{\ell(x)-1}-1,\tu y)}{2}\Bigg\rfloor 
=f\big(2^{\ell(x)-1}-1, \tu y\big) \div 2
    $$
where $\beta(\ell(z))=2^{\ell(z)}-1$ is the greater integer the length of which is $\ell(z)$, i.e.~here $2^{\ell(x)-1}-1$ is the greatest integer the length of which is $\ell(x)-1$.
Hence, starting with $x>0$, there are $\ell(x)-1$ jumps of values.

\begin{proposition}\label{prop:ODE3}
    If $f$ is defined by $\ell$-\emph{ODE}$_3$ from $g$ in $\AC^0$, then $f$ is in $\AC^0$. 
\end{proposition}
\begin{proof}
    As clarified by the remarks above, relying on Definition~\ref{def:ODE3}, and since $(a\div 2^b) \div 2 = a\div 2^{b+1}$, it is easily shown by induction that, for all $x$ and $\tu y$:
    \begin{align*}
        f(x,\tu y) &= \Bigg\lfloor \frac{g(\tu y)}{\prod^{\ell(x)-1}_{u=0} 2}\Bigg\rfloor \\
        &= \Bigg\lfloor \frac{g(\tu y)}{2^{\ell(x)}}\Bigg\rfloor \\
        &= g(\tu y) \div 2^{\ell(x)}.
    \end{align*}
    This corresponds to right-shifting $g(\tu y)$ a number of times equal to $\ell(x)$, which can be easily implemented by a constant depth circuit.
\end{proof}

\paragraph{Alternative ODE Schemas}

As a last variant, we consider a schema $\ell$-ODE$_4$, which encapsulates both left- and right-shifting (with no bit addition) simultaneously. It will be proved useful to give an alternative characterization of $\AC^0$.


\begin{defn}[Schema $\ell$-ODE$_4$]
    Given $g,k:\Nat^p\to \Nat$, the function $f:\Nat^{p+1}\to \Nat$ is defined by $\ell$-\emph{ODE}$_4$ from functions $g$ and $k$ if it is the solution of the IVP:
    \begin{align*}
        f(0,\tu y) &= g(\tu y) \\
        \frac{\partial f(x,\tu y)}{\partial \ell} &= \big(2^{\pm \ell(k(\tu y))}-1\big) \times_\uparrow f(x,\tu y)
    \end{align*}
    where $\times_\uparrow$ is so defined that $\frac{v}{2^{v'}} \times_\uparrow z = \Big\lceil \frac{v\times z}{2^{v'}}\Big\rceil$.
\end{defn}
\noindent
\begin{remark}\label{remark:ODE4}
It is clear that both wk-\emph{ODE}$_2$ and $\ell$-\emph{ODE}$_3$ are special cases of $\ell$-\emph{ODE}$_4$, in particular such that $\pm \ell(k(\tu y))= \ell(\tu y)$ and $\pm \ell(k(\tu y))= -1$, respectively.
\end{remark}

\begin{remark}\label{remark:2ellODE4}
    The IVPs presented in Example~\ref{ex:2ell} are not only instances of $\ell$-\emph{ODE}$_1$ (with $h(x,\tu y)=0$), but also of $\ell$-\emph{ODE}$_4$, considering $\frac{\partial f(x,\tu y)}{\partial \ell}=(2^{+\ell(k(\tu y))}-1) \times_\uparrow f(x,\tu y)$ with $k(\tu y)=1$.
\end{remark}

\begin{proposition}\label{prop:ODE4}
    If $f$ is defined by \emph{$\ell$-\emph{ODE}$_4$} from $g$ and $k$ in $\AC^0$,
    then $f$ is in $\AC^0$ as well.
\end{proposition}
\begin{proof}
    The proof is by cases:
    \begin{itemize}
        \itemsep0em
        \item Assume $k(\tu y)=0$.
        Then, for any $x$ and $\tu y$, $f(x,\tu y)=g(\tu y)$, which is in $\AC^0$ for hypothesis.
        \item Assume $+\ell(k(\tu y))$.
        Then, $f$ is defined from $g$ and $h$ as the solution of the IVP:
        \begin{align*}
            f(0,\tu y) &= g(\tu y) \\
            \frac{\partial f(x,\tu y)}{\partial \ell} &= \big(2^{\ell(k(\tu y))}-1\big) \times_\uparrow f(x,\tu y)
        \end{align*}
        and the proof is similar to that of Proposition~\ref{prop:ODE2}. 
        Indeed, since $\ell(k(\tu y)) >1$, we can simply treat $\times_\uparrow$ as $\times$ (where, as said, we use the standard multiplication symbol $\times$ with a slight abuse of notation, as it corresponds to multiplication by a power of 2, i.e.~corresponds to left-shifting and can be computed in constant depth).
        So for every $x$ and $\tu y$:
        $$
        f(x,\tu y) = \prod^{\ell(x)-1}_{t=0} 2^{\ell(k(\tu y))} \times g(\tu y)
        $$
        Concretely, this corresponds to left shifting the value of $g(\tu y)$ by padding $2^{\ell(k(\tu y))\times \ell(x)}$ zeros on the right and can be computed in $\AC^0$.
        \item Assume $-\ell(k(\tu y))$. 
        Then, $f$ is defined from $g$ and $h$ as the solution of the IVP below:
        \begin{align*}
            f(0,\tu y) &= g(\tu y) \\
            \frac{\partial f(x,\tu y)}{\partial \ell} &= \big(2^{-\ell(k(\tu y))}-1\big) \times_\uparrow f(x,\tu y).
        \end{align*}
        The special case of $\ell(k(\tu y))=1$ precisely corresponds to $\ell$-ODE$_3$, while the other cases are natural generalizations of it.
        Recall that $(2^{-1}-1)\times_\uparrow f(x,\tu y)= - \frac{1}{2}\times_\uparrow f(x,\tu y)= - \Big\lceil \frac{f(x,\tu y)}{2}\Big\rceil$, and we have seen that, since $(a\div 2^b)\div 2 = a\div 2^{b+1}$, it is proved by induction that, for all $x$ and $\tu y$:
        $
            f(x,\tu y) = \Big\lfloor \frac{g(\tu y)}{2^{\ell(x)}}\Big\rfloor 
            = g(\tu y) \div 2^{\ell(x)}.
        $
        This intuitively corresponds to right-shifting $g(\tu y)$ a number of times equal to $\ell(x)$ and can be computed by a constant depth circuit.
        When dealing with $\ell(k(\tu y))>1$, the procedure is analogous but we consider a sequence of right-shifting such that ``basic shifting'' is no more of one bit. 
        The number of digits for of the basic right-shifting is determined by the value of $\ell(k(\tu y))$, namely
        $f(x,\tu y) = \Big\lfloor \frac{g(\tu y)}{2^{\ell(x)\ell (y)}}\Big\rfloor$.
    \end{itemize}
\end{proof}


\subsection{An ODE Characterization of $\AC^0$}\label{sec:FAC}
We define a new class of functions, crucially relying on two of the ODE schemas just introduced:
$$
\ACDL = [\fun{0}, \fun{1},  \ell, \fun{sg}, +, -, \div 2, \pi^p_i; \circ, \ell\text{-ODE}_2, \ell\text{-ODE}_3].
$$
Observe that all its basic functions and (restricted) schemas are natural in the context of differential equations and calculus.
In $\ACDL$, multiplication is, of course, not allowed.
Compared to $\LDL$, the linear-length ODE schema is substituted by the two schemas $\ell$-ODE$_2$ and $\ell$-ODE$_3$, characterized by a very limited form of ``multiplication'' and, as seen, intuitively capturing left- and right-shifting.

In order to prove that $\ACDL$ characterizes $\AC^0$ we start by providing an \emph{indirect} proof that $\AC^0\subseteq \ACDL$.
This is established by showing that basic functions and schemas defining Clote's $\mathcal{A}_0$ (so its arithmetization of log-time bounded RAM) can be simulated in our setting by functions and schemas of $\ACDL$.
Preliminary, observe that some important operations ``come for free'' by composition.
For instance, the modulo 2 operation is defined as $x \; \fun{mod \; 2}= x-\big\lfloor \frac{x}{2}\big\rfloor - \big\lfloor \frac{x}{2}\big\rfloor$, while the binary successor functions are expressed in our setting as $\fun{s}_0(x)=x+x$ and $\fun{s}_1(x)=\fun{s}_0(x)+\fun{1}$ (being the constant $\fun{1}$ and + basic functions of $\ACDL$).

\subsubsection{The Smash Function 2$^{\ell(x)\times \ell(y)}$.}
The smash function $x\# y:x,y\mapsto 2^{\ell(x)\times \ell(y)}$ can be rewritten as the solution of the IVP below:
\begin{align*}
    f(0, y) &= 1 \\
    \frac{\partial f(x,y)}{\partial \ell} &= \big(2^{\ell(y)}-1\big) \times f(x,y).
\end{align*}
Indeed, $f(x,y)=\prod^{\ell(x)-1}_{t=0} 2^{\ell (y)} \times 1 = 2^{\ell(y)\times \ell(x)}$.
This is clearly an instance of $\ell$-ODE$_2$ (and of $wk$-ODE$_2$), such that $g(\tu y)=1$, $k(\tu y)=y$ and $h(x,\tu y)=0$.
Recall that, since $h(x,\tu y)=0$, also the case $y=0$ is properly captured.

\begin{remark}\label{remark:smashODE4}
    Observe that the IVP corresponding to the smash function is also an instance of $\ell$-\emph{ODE}$_4$ (namely, such that $g(\tu y)=1$, $\pm\ell(k(\tu y))= + \ell(k(y))$ and $k(\tu y)=y$).
\end{remark}

\subsubsection{The $\fun{BIT}$ Function}\label{sec:BIT}
Intuitively, the function $\fun{BIT}(x,y)$ returns the $y^{th}$ bit in the binary representation of $x$.
In order to capture it, a series of auxiliary functions are needed:
\begin{itemize}
    \itemsep0em
    \item the \emph{log most significant part function} $\fun{msp}(x,y):x,y\mapsto \Big\lfloor \frac{y}{2^{\ell(x)}}\Big\rfloor$, which can be rewritten via $\ell$-ODE$_3$.
    \item the \emph{basic conditional function} $\fun{if}(x,y,z)$, returning $y$ if $x=0$ and $z$ otherwise, can be rewritten in our setting by composition, using, in particular, the ``shift function'' $2^{\ell(x)}\times y$ (as seen, defined in $\ACDL$ using $\ell$-ODE$_1$).
    \item the \emph{special bit function} $\fun{bit}(x,y)$, returning 1 when the $\ell(y)^{th}$ bit of $x$ is 1, can be rewritten in $\ACDL$ due to $\fun{msp}$.
    \item the \emph{bounded exponentiation function} $\fun{bexp}(x,y)$, that, for any $y\leq \ell(x)$, returns $2^y$, can be obtained relying on functions in $\ACDL$, including, in particular $\fun{msp}, \fun{if}$ and $2^{\ell(\cdot)}$ together with ODE schemas.
\end{itemize}
Then, using $\fun{bexp}$ and $\fun{bit}$, the desired function $\fun{BIT}$ can be rewritten in our setting by composition:
$$
\fun{BIT}(x,y) = \fun{bit}\big(x, \fun{bexp}(x,y) -1 \big).
$$
Notably, alternative proofs are possible, but the one proposed here, and based on the introduction of $\fun{bexp}$, not only has the advantage of being straightforward, but also avoids the unnatural use of the most significant part function, $\fun{MSP}$.

\paragraph{Log most significant part function}
The \emph{log most significant part function} $\fun{msp}(x,y):x,y \mapsto \big\lfloor \frac{y}{2^{\ell(x)}}\big\rfloor$ can be rewritten as the solution of the IVP below:
\begin{align*}
    f(0,y) &= y \\
    \frac{f(x,y)}{\partial \ell} &= - \Bigg\lceil \frac{f(x,y)}{2} \Bigg\rceil
\end{align*}
which is clearly an instance of $\ell$-ODE$_3$, such that $g(\tu y)=y$.

\begin{remark}\label{remark:mspODE}
    This is not only an instance of $\ell$-\emph{ODE}$_3$, but also of $\ell$-\emph{ODE}$_4$, such that $g(\tu y)=y$, $\pm \ell(k(\tu y))=-\ell(k(\tu y))$ and $k(\tu y)=1$.
\end{remark}

\paragraph{Basic conditional function}
We introduce the basic conditional function:
$$
\fun{if}(x,y,z)= \begin{cases}
    y \quad \quad &\text{if } x= 0 \\
    z \quad \quad &\text{otherwise.}
\end{cases}
$$
Notice that this function is also crucial to rewrite the CRN schema.
As seen, the ``shift function'' $2^{\ell(x)} \times y$ can be rewritten via $\ell$-ODE$_1$ (see Example~\ref{ex:2ell}).
Thus, $\fun{if}(x,y,z)$ is simulated in our setting by composition from shift, addition and subtraction:
$$
\fun{if}(x,y,z) = \big(2^{\ell(1-\fun{sg}(x))} \times y - y\big) + \big(2^{\ell(\fun{sg}(x))}\times z-z\big).
$$
Indeed, as desired, if $x=0$, then $\fun{sg}(x)=0$ and $\ell(\fun{sg}(x))=0$, so that $\fun{if}(0,y,z)=\big(2^{\ell(1)} \times y - y\big) + \big(2^{\ell(0)}\times z-z\big)=y$;
similarly, for $x\neq 0$, $\fun{if}(0,y,z)=\big(2^{\ell(0)}\times y - y\big) + \big(2^{\ell(1)} \times z - z\big) =z$.
Generalizing this definition we can capture the more general conditional function below:
$$
\fun{cond}(x,v,y,z) = \begin{cases}
    y \quad \quad &\text{if } x<v \\
    z \quad \quad &\text{otherwise.}
\end{cases}
$$


\paragraph{Special bit function}
Then, we consider the special bit function $\fun{bit}(x,y)$ returning 1 when the $\ell(y)^{th}$ bit of $x$ is 1.
This can be rewritten in $\ACDL$ due to $\fun{msp}$:
$$
\fun{bit}(x,y) = \fun{msp}(y,x) - 2 \times \fun{msp}(2y +1,x ).
$$

\paragraph{Bounded exponentiation function}
Finally, we introduce the function $\fun{bexp}(x,y)$, that, for any $y\leq \ell(x)$, returns $2^y$.
We start by defining $f_{aux}(t,x,i)$ by the $\ell$-ODE$_1$ schema below:
\begin{align*}
    f_{aux}(\fun{0},x,i) &= \fun{if}(i,\fun{1},\fun{0}) \\
    \frac{\partial f_{aux}(t,x,i)}{\partial \ell(t)} &= f_{aux}(t,x,i) + h_{aux}(t,i)
\end{align*}
where
$$
h_{aux}(t,i) = \fun{if}(\ell(t)-i,\fun{1},\fun{0})
$$
Observe that, as seen, $\fun{if}$ can be rewritten in $\ACDL$ (while $\ell$ and subtraction are basic functions).
Then, for $i\leq \ell(x)$, we obtain $f_{aux}(x,x,i)=2^{\ell(x)-i}$.
The function $\fun{bexp}$ is then defined as follows:
$$
\fun{bexp}(x,i) = \fun{msp}\big(f_{aux}(x,x,i)-\fun{1}, 2^{\ell(x)}\big) = \Bigg\lfloor \frac{2^{\ell(x)}}{2^{\ell(x)-i}}\Bigg\rfloor = 2^i
$$
since the length of $f_{aux}(x,x,i)-1$ is $\ell(x)-i$.
Clearly, the function $\fun{bexp}$ is also in $\ACDL$, as all the functions involved in its definitions (namely, $\fun{msp}$, $f_{aux}$ and $2^{\ell(\cdot)})$ are in $\ACDL$.

\subsubsection{The CRN Schema}\label{sec:CRN}
A function $f$ defined by CRN from $g, h_0$ and $h_1$ can be simulated in $\ACDL$ via the $\ell$-ODE$_1$ schema.
Let us consider the following IVP:
\begin{align*}
    F(0,x,\tu y) &= g(\tu y) \\
    \frac{\partial F(t,x,\tu y)}{\partial \ell(t)} &= F(t,x,\tu y) + h(t+1, x, \tu y)
\end{align*}
where $h(t,x,\tu y)\in \{0,1\}$ is, in turn, defined as:
$$
\fun{if} \Big(\fun{bit}\big(x, 2^{\ell(x)-\ell(t)}-1\big),
h_0\big(\fun{msp}(2^{\ell(x)-\ell(t)},x),\tu y\big), h_1\big(\fun{msp}(2^{\ell(x)-\ell(t)}, x), \tu y\big)\Big).
$$
The function $F(t,x,\tu y)$ is clearly an instance of $\ell$-ODE$_1$ and $h(t,x,\tu y)$ is defined by composition from functions proved to be in $\ACDL$.
Then, we set $f(x,\tu y)=F(x,x,\tu y)$.

\subsubsection{Characterizing $\AC^0$}

We now have all the ingredients to prove our main result.

\begin{theorem}\label{theorem:ACDL}
    $\ACDL=\AC^0$.
\end{theorem}
\begin{proof}
    $\ACDL \subseteq \AC^0$. All basic functions of $\ACDL$ are computable in $\AC^0$.
    Moreover, the class is closed under composition and, by Propositions~\ref{prop:ODE2} and~\ref{prop:ODE3}, under $\ell$-ODE$_2$ and $\ell$-ODE$_3$.
\\
    $\AC^0\subseteq \ACDL$. As we have just proved, all functions and schemas constituting $\mathcal{A}_0$ can be rewritten in $\ACDL$.
    Then, it is routine to mimic Clote's arithmetization~\cite{Clote88,Clote1990} and obtain a direct proof in our setting. 
\end{proof}


\noindent
A careful analysis of the above definitions shows that the full power of $\ell$-ODE$_2$ is actually used only to capture the smash function $\#$.
Therefore, a class equivalent to $\ACDL$ is defined by allowing $\#$ and by replacing $\ell$-ODE$_2$ with the simpler $\ell$-ODE$_1$ schema.

\begin{lemma}\label{cor:FAC}
    $\AC^0 = [\fun{0}, \fun{1}, \ell, \fun{sg}, +, -, \div 2, \#, \pi^p_i; \circ, \ell\emph{-ODE}_1, \ell\emph{-ODE}_3].$
\end{lemma}

To further clarify what features of discrete ODEs are really essential for capturing $\AC^0$ computation, we present alternative characterizations for this class, obtained by changing either the ODE schemas or the basic functions.


\begin{lemma}
    $
    \AC^0=[\fun{0}, \fun{1}, \ell, \fun{sg}, +, -, \div 2, \pi^p_i ; \circ, \ell\emph{-ODE}_1, wk\emph{-ODE}_2, \ell\emph{-ODE}_3].
    $
\end{lemma}

\begin{proof}
    Right-to-left inclusion:  As for Theorem~\ref{theorem:ACDL}, we know that all basic functions are computable in $\AC^0$ and that the class is closed under $\ell$-ODE$_1$ (Proposition~\ref{prop:ODE1}), $wk$-ODE$_2$ (special case of Proposition~\ref{prop:ODE2}) and $\ell$-ODE$_3$ (Proposition~\ref{prop:ODE3}).
    \\
    Left-to-write inclusion: As for Theorem~\ref{theorem:ACDL} the proof is indirect and obtained by showing that we can rewrite all basic functions and schemas of $\mathcal{A}_0$ in the given function algebra: 
    \begin{itemize}
        \itemsep0em
        \item the binary successor functions $\fun{s}_0$ and $\fun{s}_1$ can be easily rewritten by addition and composition.
        \item the smash function $\#$ can be rewritten via $wk$-ODE$_2$; indeed $x\#y$ is the solution of the IVP below:
        \begin{align*}
            f(0,y) &= 1 \\
            \frac{\partial f(x,y)}{\partial \ell} &= (2^{\ell(y)}-1)\times f(x)
        \end{align*}
        %
        \item the function $\fun{BIT}$ is rewritten as  in Section~\ref{sec:BIT}, namely by composition from $\fun{bit}$ and $\fun{bexp}$.
        Indeed, we have shown that $\fun{bit}$ is defined by composition using $\fun{msp}$, in turn rewritten using $\ell$-ODE$_3$, while $\fun{bexp}$ is defined relying not only on $\fun{msp}$ but also on $f_{aux}$ and $2^{\ell(\cdot)}$, which are both defined by $\ell$-ODE$_1$, in the former case passing through $\fun{if}$.
        \item the CRN schema can be rewritten  as done in Section~\ref{sec:CRN}, i.e.~by $\ell$-ODE$_1$ and functions definable in the class  ($2^{\ell(\cdot)}, \fun{bit}$ and $\fun{msp}$).
    \end{itemize}
\end{proof}

By similar arguments and by remembering that $\AC^0$ is closed under the $\ell$-ODE$_4$ schema (Proposition~\ref{prop:ODE4}), one can also prove the following corollary.


\begin{lemma}
    $\AC^0= [\fun{0}, \fun{1}, \ell, \fun{sg}, +, -, \div 2, \pi^p_i; \circ, \ell\emph{-ODE}_1, \ell\emph{-ODE}_4].$
\end{lemma}
%
    %
    %

\subsection{An ODE Characterization of $\FTC^0$}\label{sec:FTC}

As an additional consequence, an ODE-characterization for $\FTC^0$ is naturally obtained, this time passing through $\mathcal{TC}_0$~\cite{CloteTakeuti}.
Specifically, we consider the following algebra, obtained by endowing $\ACDL$ with the basic function $\times$:
$$
\TCDL = [\fun{0}, \fun{1}, \ell, \fun{sg}, +, -, \div 2,\times, \pi^p_i; \circ, \ell\text{-ODE}_2, \ell\text{-ODE}_3].
$$
Again, the desired closure properties are shown to hold.
\begin{proposition}\label{prop:ODE2TC}
    Let $f$ be defined by $\ell$-\emph{ODE}$_2$ from functions in $\FTC^0$.
    Then, $f$ is in $\FTC^0$ as well.
\end{proposition}
\begin{proof}
    The proof is similar to that of Proposition~\ref{prop:ODE2}.
    The main difference concerns $g(\tu y)$ and $h(x,\tu y)$, which are now expressions possibly including $\times$.
    
    As seen, for all $x$ and $\tu y$:
    $$
    f(x,\tu y) = \sum^{\ell(x)-1}_{u=-1} \Bigg(\prod^{\ell(x)-1}_{t=u+1} 2^{\ell(k(\tu y))}\Bigg) \times h\big(\alpha(u), \tu y\big)
    $$
    with the convention that $\alpha(u)=2^u-1$, $\prod^{x-1}_x \kappa(x)=1$ and $h(\alpha(-1), \tu y) = f(0,\tu y)$.
    Intuitively, the (constant-depth) circuit we are going to construct is analogous to that of Proposition~\ref{prop:ODE2}, but the values to be initially computed in parallel are obtained even via $\times$.
    However, the introduction of multiplication does not affect the overall structure of the circuit, as $h(\alpha(u), \tu y) \in \{0,1\}$, so that the final sum again corresponds to a simple bit-concatenation (without carries).

    More precisely, the desired constant-depth circuit is defined as follows:
    \begin{itemize}
        \itemsep0em
        \item In parallel, compute the values of $g(\tu y)$ and, for any $u=0,\dots, \ell(x)-1$, of $h(\alpha(u), \tu y)$.
        Observe that this can be done in $\FTC^0$, but possibly not in $\AC^0$, as now arithmetic expressions may include $\times$.
        \item The value of $g(\tu y)$ is shifted by padding $2^{\ell(k(\tu y)) \times \ell(x)}$ zeros on the right and, for $u\ge 0$, all values $h(\alpha(u), \tu y)$ are shifted by padding on the right $\ell(k(\tu y))$ zeros $\ell(x)-u-1$ times, and by padding on the left $\ell(k(\tu y))$ zeros $\ell(g(\tu y))+u$ times.
        \item Compute bit by bit the disjunction of the above values.
    \end{itemize}
\end{proof}

\begin{proposition}\label{prop:ODE3TC}
    Let $f$ be defined by $\ell$-\emph{ODE}$_3$ from functions in $\FTC^0$. Then, $f$ is in $\FTC^0$ as well.
\end{proposition}
\begin{proof}
    Again, the proof is a straightforward generalization of Proposition~\ref{prop:ODE3}, with the same provisos of Proposition~\ref{prop:ODE2TC}, namely considering that computation of level 0, i.e.~the computation of the initial values $g(\tu y)$ and $h(x, \tu y)$ as expressions possibly including $\times$, does not affect the overall structure of the circuit.
\end{proof}
\noindent
Then, the desired characterization straightforwardly follows from Propositions~\ref{prop:ODE2TC} and~\ref{prop:ODE3TC} and from the encoding of the proof by~\cite{CloteTakeuti} in our ODE setting.

\begin{theorem}\label{theorem:TCDL}
    $\TCDL=\FTC^0$.
\end{theorem}

An alternative characterization of $\FTC^0$ is obtained by considering the following generalized version of $\ell$-ODE$_2$, called $\ell$-ODE$_2^*$, with no constraint over the function $k$.

\begin{defn}[Schema $\ell$-ODE$_2^*$]\label{defn:ODEstar}
    Let $g:\Nat^p\to \Nat, h:\Nat^{p+1}\to \Nat$ and $k:\Nat^p\to \Nat$, where $h$ takes values in $\{0,1\}$.
    Then, the function $f:\Nat^{p+1} \to \Nat$ is defined by $\ell$-\emph{ODE}$^*_2$ from $g,h$ and $k$ when it is the solution of the IVP below:
    \begin{align*}
    f(0,\tu y) &= g(\tu y) \\
    \frac{\partial f(x,\tu y)}{\partial \ell} &= \big(2^{\ell(k(\tu y))}-1\big) \times f(x,\tu y) + h(x,\tu y).
    \end{align*}
\end{defn}
This schema is really more expressive than $\ell$-ODE$_2$ as shown by the following example, presenting an encoding by $\ell$-ODE$_2^*$ of the binary counting function, which is not in $\AC^0$ (see~\cite{Vollmer}).

\begin{example}[$\fun{bcount}$]\label{ex:bcount}
    Observe that if $k(\tu y)=0$, then $\ell$-\emph{ODE}$_2^*$ is enough to express the binary counting function $\fun{bcount}(x)$, that outputs the sum of the bits of $x$.
    Indeed, $\fun{bcount}(x)=f(x,x)$ where $f$ is the solution of the IVP below:
    \begin{align*}
    f(0,\tu y) &= \fun{BIT}_0(\tu y) \\
    \frac{\partial f(x,\tu y)}{\partial \ell} &= \fun{BIT}(\tu y, x),
    \end{align*}
    such that $\fun{BIT}_0(z)$ is a function returning the 0th bit of $z$ (and can be rewritten in $\mathbb{TCDL}^*$).
    This system is clearly an instance of $\ell$-\emph{ODE}$_2^*$, such that $g(\tu y)=\fun{BIT}_0(\tu y)$, $k(\tu y)=0$ and $h=\fun{BIT}$.
\end{example}
\noindent
Therefore, together with the function $\fun{BIT}$, $\ell$-ODE$_2^*$ is enough to capture majority computation and to ``simulate'' multiplication, see~\cite{Vollmer}, which is the essential step to show that $\TCDL \subseteq \TCDL^*$.
In addition, $\FTC^0$ is closed under $\ell$-ODE$_2^*$.

\begin{proposition}\label{prop:ODE2sTC}
    Let $f$ be defined by $\ell$-\emph{ODE}$_2^*$ from functions in $\FTC^0$. Then, $f$ is in $\FTC^0$ as well. 
\end{proposition}
\begin{proof}[Proof Sketch]
    The proof is by cases. 
    \begin{itemize}
        \itemsep0em
        \item If $k(\tu y)\neq 0$, then the proof is precisely as Proposition~\ref{prop:ODE2TC}.
        \item If $k(\tu y)=0$, then computing $f$ corresponds to binary counting (plus a single addition), which is in $\FTC^0$~\cite{Vollmer}.
    \end{itemize}
\end{proof}

\noindent
This observation, together with the fact that what really makes $\ell$-ODE$^*_2$ more expressive than $\ell$-ODE$_2$ is its behavior for $k(\tu y)=0$, leads us to an alternative characterization for $\FTC^0$, in line with the one of Corollary~\ref{cor:FAC}:
$$
\TCDL^* = [\fun{0}, \fun{1}, \ell, \fun{sg}, +, -, \div 2, \pi^p_i; \circ, \ell\text{-ODE}_2^*, \ell\text{-ODE}_3].
$$
Clearly $\TCDL^*$ does not include $\times$ as a basic function.

\begin{corollary}
    $\FTC^0=\TCDL^*.$
\end{corollary}
\begin{proof}[Proof Sketch]
    $\mathbb{TCDL}^* \subseteq \FTC^0$. It is a straightforward consequence of Proposition~\ref{prop:ODE2sTC}.
    \\
    $\FTC^0\subseteq \mathbb{TCDL}^*$. The proof is indirect and passes through Theorem~\ref{theorem:TCDL}. In particular, we only need to prove that $\TCDL \subseteq \TCDL^*$, i.e. that $\times$ is expressible in $\TCDL^*$, which follows from Example~\ref{ex:bcount} and~\cite{Vollmer}.
    \\
\end{proof}

Let us consider the following stronger version of $\ell$-ODE$_1$.

\begin{defn}[Schema $\ell$-ODE$_1^*$]\label{defn:ODEstar}
    Let $g:\Nat^p \to \Nat$ and $h,k:\Nat^{p+1}\to \{0,1\}$.
    Then, the function $f:\Nat^{p+1} \to \Nat$ is defined by $\ell$-\emph{ODE}$_1^*$ from $g,h$ and $k$, if it is the solution of the IVP:
    \begin{align*}
        f(0,\tu y) &= g(\tu y) \\
        \frac{\partial f(x,\tu y)}{\partial \ell} &= k(x,\tu y) \times f(x,\tu y) + h(x,\tu y).
    \end{align*}
\end{defn}
\noindent
By straightforward inspection and Example~\ref{ex:bcount}, it is also proved that $\FTC^0$ is again captured by adding $\#$ to the set of basic functions of $\TCDL^*$ and by replacing $\ell$-ODE$_2^*$ with the simpler schema $\ell$-ODE$_1^*$.

\begin{corollary}
    $\FTC^0=[\fun{0},\fun{1}, \ell, \fun{sg}, +, -, \div 2, \#, \pi^p_i; \circ, \ell\text{-}\emph{ODE}_1^*, \ell\text{-}\emph{ODE}_3].$
\end{corollary}

\subsection{Alternative Direct Proofs}\label{sec:direct}

In this Section, we introduce alternative parameterized classes $\ACDL_{\mathcal{C}}$ and $\TCDL_{\mathcal{C}}$, respectively extending $\ACDL$ and $\TCDL$ with new basic functions that provide a description of  circuit families $\mathcal{C}=(C_n)_{n\ge 0}$ of polynomial size and constant depth  used for computation.
In this non-uniform context, we prove by a direct argument (not referencing to results in ~\cite{Clote1990,CloteTakeuti} or any other previous characterizations) that if a function is computable by such family $\mathcal{C}$ (resp. with majority gates in addition), then it is definable in $\ACDL_{\mathcal{C}}$ (Section~\ref{sec:directACDL}) (resp. $\TCDL_{\mathcal{C}}$; see Section~\ref{sec:directTCDL}).

\subsubsection{Direct Completeness for $\ACDL_{\mathcal{C}}$}\label{sec:directACDL}

Let $\mathcal{C}=(C_n)_{n\ge 0}$ be a class of circuits of polynomial size $n^k$, for some $k\in \Nat$, and constant depth $d$.
We assume that each circuit $C_n$ is in a special normal form, such that it strictly alternates between $\wedge$ and $\vee$ (and edges are only between gates of consecutive layers): input gates are all at level 0 as well as negation gates (by convention), at even levels all gates are $\wedge$ gates, at odd levels all gates are $\vee$ gates, and the depth of the circuit $d$ is even (so output gates are $\wedge$ gates). Hence, such a circuit can be described by an acyclic directed graph on a vertex set of size $n^k$ with vertices of in-degree (resp. out-degree) $0$ representing the input (resp. output) gates.


In this context we keep all the basic functions of $\ACDL$, but we add a set $\mathbf{circ}_{\mathcal{C}}=\{C, L_0^{in}, L_0^{\neg}, L_1,...,L_d,m \}$ with  $L_0^{in}, L_0^{\neg}, L_e \subseteq \Nat^2$, for $e\in \{1,\dots,d\}$,  $C \subseteq\Nat^3$ (these predicates are seen as their characteristic functions) and $m:\Nat\rightarrow\Nat$. The predicates and functions in $\mathbf{circ}_{\mathcal{C}}$ will encode the circuit family.
Intuitively, $m(\ell(x))$ (denoted $m$ if there is no ambiguity) is equal to the number of output gates of the circuit $C_{\ell(x)}$, while $L_0^{in}, L_0^{\neg}, L_e$ describe the level of gates (and, implicitly, their type): $L_0^{in}$ refers to input gates, $L_0^{\neg}$ to negation (of input) gates, and $L_e$ to $\wedge$ and $\vee$ gates, depending on $e$ being odd or even.
More precisely, for integers $a, x$ with $a\leq \ell(x)$, $L_0^{in}(a,x)$ holds if the $a^{th}$ input gate of the circuit $C_{\ell(x)}$ with input $x$  (or, equivalently, the $a^{th}$ bit of $x$) is $1$. Similarly, for $a, x$ with $a\leq \ell(x)^k$, $L_e(a,x)$ holds if the $a^{th}$ gate of circuit $C_{\ell(x)}$ is at level $e$. 
Finally, the predicate $C$ will describe the underlying graph of the circuit:
for any integers $x, a,b$ with $a,b\leq \ell(x)^k$, it holds that $(x,a,b) \in C$ when in $C_{\ell(x)}$ the $a^{th}$ gate of some level is a predecessor of the $b^{th}$ gate of the next level.
Thus, in this encoding, $a$ and $b$ are exponentially smaller than $x$.
%
%
Since we aim to define functions over integers, to simplify unessential technical details, we assume that input gates are numbered from $0$ to $n-1$, negation of input gates  from $n$ to $2n-1$ and  output gates are indexed from $n^k-m$ to $n^k-1$ with $m$ being the number of output gates.
%
By considering the functions corresponding to the given relations, we obtain the desired ODE style family of classes (parametrized on $\mathcal{C}$) as follows:
$$
\ACDL_{\mathcal{C}} = [\fun{0}, \fun{1}, \mathbf{circ}_{\mathcal{C}}, \fun{sg}, \ell, +, -, \div 2, \pi^p_i; \circ, \ell\text{-ODE}_2, \ell\text{-ODE}_3].
$$

\begin{remark}
    If the family $\mathcal{C}$ is $\Dlogtime$-uniform, then the functions in $\mathbf{circ}_{\mathcal{C}}$ are computable in $\mathcal{A}_0$.
    Consequently, in this case, it holds that $\ACDL_{\mathcal{C}}=\ACDL$.
\end{remark}

In this non-necessarily uniform setting, a completeness proof still holds.

\begin{proposition}\label{prop:dirACDL}
    If a function $f:\Nat \to \Nat$ is computable by a family $\mathcal{C}=(C_n)_{n\ge 0}$ of polynomial size and constant depth circuits, then it is in $\ACDL_{\mathcal{C}}$.
\end{proposition}
\noindent
To prove it we need the following Lemma, which is a consequence of Remark~\ref{remark:boundedQ}.

\begin{lemma}\label{lemma:min}
    Let $g$ and $h$ be functions computable in $\ACDL_{\mathcal{C}}$, $k\in \Nat$ and $\tu x = x_1, \dots, x_h$. 
    Then, the function $\fun{min}_{i\leq \ell(x_1)^k} \{g(i,\tu x) : h(i,\tu x) \triangleright j\}$, for $\triangleright \in \{<,\leq, >, \ge, =\}$ and $j\in \{0,1\}$ is in $\ACDL_{\mathcal{C}}$.
\end{lemma}


\begin{proof}[Proof of Proposition~\ref{prop:dirACDL}]
Let $\fun{Eval}(t,x)$ be a function that returns  the value of the $t^{th}$ output gate of the circuit $C_{\ell(x)}$ of input $x$ when $n^k-m\leq t\leq n^k-1$ and 0 otherwise (other values are indifferent).
Then, the following expression defines a function $f$ such that $f(\alpha(m)+1,x)$, with $\alpha(m)=2^{m}-1$ (see Notation~\ref{notation: alpha functions}) outputs the value of the computation of $C_n$ on input $x$, with $n=\ell(x), m=m(\ell(x))$:
$$
\frac{\partial f(y,x)}{\partial \ell(y)} = f(y,x) + \fun{Eval}\big(n^k-1 -m + \ell(y),x\big)
$$
with $f(1,x)=\fun{Eval}(n^k-m,x)$.
Intuitively, the function above computes the successive suffixes of the output word, starting from the bits of bigger weights~\footnote{Indeed, $\alpha(m)+1=2^{m}$ is the first integer of length $m$+1. Hence, $f(2^m,x)=2f(\alpha(m),x)+ \fun{Eval}(n^k-1,x)=4f(\alpha(m-1),x)+ 2\fun{Eval}(n^k-2,x) + \fun{Eval}(n^k-1,x)$, etc is the number whose binary decomposition is given by the sequence $\fun{Eval}(n^k-m,x), \fun{Eval}(n^k-m+1,x),...,\fun{Eval}(n^k-2,x), \fun{Eval}(n^k-1,x) $ 
}.
Remarkably, this is an instance of the ODE$_1$ schema (indeed, $\fun{Eval}(y,x)\in \{0,1\}$).
So, the given $f$ can be rewritten in $\ACDL_{\mathcal{C}}$.

It remains to describe how the function $\fun{Eval}(t,x)$ is defined.
Again, we assume that $C_n$ has depth $d$, that it is in normal form, and that $d$ is even. 
Concretely, we start by considering a special (bounded) minimum operator function defined as in Lemma~\ref{lemma:min} and such that $k\in \Nat$ and
$h(t,\tu x)\in \{0,1\}$ for $t\in \Nat$ and $\tu x=x_1,\dots, x_h$.
Intuitively, given $i\in\{0,\dots, \ell(x_1)^k\}$, this function computes the minimum of the values of $g(i,\tu x)$, for $i$ and $\tu x$ such that $h(i,\tu x)\triangleright j$.

The inductive definition of $\fun{Eval}$ relies on those of the $d+1$ functions $\fun{Eval}_0, \dots, \fun{Eval}_d$, with $\fun{Eval}_d=\fun{Eval}$:
\begin{itemize}
    \itemsep0em
    \item $\fun{Eval}_0(t,x)$ is equal to $\fun{BIT}(t,x)$ if $L_0^{in}(t,x)$ holds and to $1-\fun{BIT}(t - \ell(x),x)$ if $L^\neg_0(t,x)$ does.
    For $t$ not corresponding to gate index, $\fun{Eval}_0(t,x)$ is set to an arbitrary value, say 0.
    Since $\fun{BIT}$ can be rewritten already in $\ACDL$ (Section~\ref{sec:BIT}), $\fun{Eval}_0$ is in $\ACDL_{\mathcal{C}}$ as well.
    \item $\fun{Eval}_{2e}(t,x)$ is equal to $\fun{min}_{i\leq \ell(x)^k}\{\fun{Eval}_{2e-1}(i,x) : C(x,i,t)=1\}$, for $L_{2e}(t,x)$, i.e.~if $t$ is the index of a gate at this level.
    The evaluation for the $i^{th}$ gate at level $2e$ (a $\wedge$-gate) is the minimum of the evaluations of its predecessors gates of level $2e-1$.
    As seen, $\fun{min}$ is in $\ACDL_{\mathcal{C}}$ (Lemma~\ref{lemma:min}), so $\fun{Eval}_{2e}$ can also be rewritten in this class.
    \item Similarly, $\fun{Eval}_{2e+1}(t,x)$ is the $1-\fun{min}_{i\leq \ell(x)^k}\{1-\fun{Eval}_{2e}(i,x) : C(x,i,t) =1\}$.
    The evaluation for the $t^{th}$ gate of level $2e+1$ (a $\vee$-gate) is the maximum among evaluations of its predecessor gates of level $2e$.
    As for $\fun{Eval}_{2e}$, $\fun{Eval}_{2e+1}$ can be rewritten in $\ACDL_{\mathcal{C}}$.
\end{itemize}
\end{proof}

\begin{remark}
    By inspecting the proofs of Propositions~\ref{prop:ODE2} and~\ref{prop:ODE3}, even the following converse to Proposition~\ref{prop:dirACDL} is proved to hold:
    If a function $f:\Nat \to \Nat$ is in $\ACDL_{\mathcal{C}}$ for some family of polynomial size and constant depth circuits $\mathcal{C}$, then there is a family of polynomial size and constant depth circuits $\mathcal{C}'$ that computes it.
\end{remark}

\subsubsection{Direct Completeness for $\TCDL_{\mathcal{C}}$}\label{sec:directTCDL}

We now consider a similar, direct characterization for $\FTC^0$.
Suppose that $C_n$ strictly alternates between $\wedge, \vee$ and $\maj$ gates, and that input gates and their negations are all at level 0, $\vee$ gates are at level $3e+1$, $\wedge$ gates are at levels $3e+2$ and $\maj$ gates are at level $3e$.
Accordingly, in this case $L_e$ describes the level of a gate of a type not limited to $\wedge, \vee$, but including $\maj$.
Then, the desired family of classes is defined as follows:
$$
\TCDL_{\mathcal{C}} = [\fun{0}, \fun{1}, \mathbf{circ}_{\mathcal{C}}, \div 2, \fun{sg}, \ell, +, -, \pi^p_i; \circ, \ell\text{-ODE}_2^*, \ell\text{-ODE}_3]
$$
where, with a slight abuse of notation, we use $\mathbf{circ}_{\mathcal{C}}$ to denote the function corresponding to a (set of) relation(s), this time including the \emph{extended L$_e$}.
The proof that non-uniform $\FTC^0\subseteq \TCDL_\mathcal{C}$ is similar to the one from Section~\ref{sec:directACDL}.

\begin{proposition}
    If a function $f:\Nat \to \Nat$ is computable by a family $\mathcal{C}=(C_n)_{n\ge 0}$ of polynomial size and constant depth circuits including $\maj$ gates, then it is in $\TCDL_{\mathcal{C}}$.
\end{proposition}

\begin{proof}
    The functions $f$ and $\fun{Eval}$ are globally defined as the corresponding ones in Section~\ref{sec:directACDL}.
    Modifications only affect the definition of $\fun{Eval}_d$ and, in particular, the inductive levels corresponding to $\maj$.
    Specifically, in this context, it is defined as follows:
    \begin{itemize}
        \itemsep0em
        \item For a given function $h$ and integer $k$, $\fun{bcount}_h(t,x) = \sum_{i\leq \ell(x)^k}h(i,t,x)$.
        Notice that this function is in $\TCDL_{\mathcal{C}}$, since it can be rewritten as an instance of $\ell$-ODE$_2^*$, and $h$ is in $\TCDL^*$ (due to Lemma~\ref{lemma:min}).
        \item For any $i$ such that $L_{3e-1}(i)$ and $3e-1 <d$:
        \begin{align*}
            v_{3e-1}^0 (i,t,x) &= \fun{sg}\big(C(x,i,t)\big) \\
            v_{3e-1}^1 (i,t,x) &= \fun{if}\big(C(x,i,t), \fun{Eval}_{3e-1(i,x)},0\big).
        \end{align*}
        The value of $v_{3e-1}^0(i,t,x)$ is 1 when $i$ is a predecessor of gate $t$ and $v_{3e-1}^1(i,t,x)$ is 1 when, in addition, the value of gate $i$, on input $x$, is 1.
        \begin{sloppypar}
        \item For $t$ such that $L_{3e}(t)$, $\fun{Eval}_{3e}(t,x)$ is defined as $\fun{sg}(\fun{bcount}_{v^0_{3e-1}}(t,x)-2\times \fun{bcount}_{v^1_{3e-1}}(t,x))$.
        The function outputs 1 when more than half of the inputs of gate $t$ are 1.
        \end{sloppypar}
    \end{itemize}
\end{proof}

\section{Conclusion}\label{sec:conclusion}

In this paper, we introduce what, to the best of our knowledge, are the first implicit characterizations for small circuit classes based on discrete ODEs. By imposing various restrictions on linearity, we obtain straightforward characterizations for both $\AC^0$ and $\FTC^0$. 
The striking simplicity of our schemas and the fact that we can pass from the former class to the latter by simply removing a `syntactical' constraint, i.e.~switching from $\ell$-ODE$_2$ to $\ell$-ODE$_2^*$, suggests that these characterizations could clarify the computational principles underlying the classes considered and enhance our understanding of the relationship between circuit computation and the ODEs machinery. 
Indeed, it emerges that the features defining circuit computation and their basic steps are mirrored by the shape and constraints on (linear) equations defining our schemas:
for example, when $A$ is a power of two, computation corresponds to shifting (the binary representation of) a number, while $B$ being restricted to values 0 or 1 ensures that addition involves only bits, thereby preventing any carry operations; additionally, when $A$ is allowed to take value 0, binary counting – which is not in $\AC^0$ but is characteristic of counting or majority gates of $\mathbf{FTC}^0$ circuits – is captured.
Notably, this framework is expected to generalize to the investigation of other relevant circuit classes.

For this reason, we view these characterizations 
as a first step in a broader project aimed at capturing additional small circuit classes.
Indeed, given the novelty of this approach, there are several promising directions for future research. 
As mentioned, since the restrictions we have applied to linear-length ODE schemas are surprisingly natural, we believe that a similar analysis could also capture computation corresponding to $k$-BRN and WBRN~\cite{Clote88,Clote1990,CloteKranakis}, paving the way for \emph{uniform} characterizations of the $\AC$ and $\FNC$ hierarchies through the lens of ODEs. Investigating intermediate classes, such as $\mathbf{FACC(2)}$ (i.e., $\AC^0$ with modulo $2$ gates), could provide new insight into the relationship between small circuit classes (with varying notions of uniformity) and counting complexity. Another intriguing avenue for exploration is the generalization of our approach to the continuous setting, specifically the analysis of small circuit classes over the reals, following the path outlined by~\cite{BlancBournez22,BlancBournez23} for $\FPTime$ and $\mathbf{FPSPACE}$. A challenging direction for future research would also be to develop logical and proof-theoretical counterparts to ODE-style algebras, such as introducing \emph{natural} rule systems (guided by the ODE design) to syntactically characterize the corresponding classes.

%
%
%
%

\section*{Acknowledgements}

Since 2023, the research of Melissa Antonelli has been funded by the Helsinki Institute for Information Technology.
The research of Arnaud Durand is partially funded by ANR Grant $\partial$ifference.
Melissa Antonelli and Arnaud Durand thank Maupertuis’ SMR Programme, on behalf of
the Institute Français in Helsinki, the French Embassy to Finland, the French Ministry of Higher
Education, Research and Innovation in partnership with the Finnish Society for Science and Letters
and the Finnish Academy of Sciences and Letters, for the financial support.

%

\bibliographystyle{plain}
\bibliography{references.bib}

\end{document}